\documentclass[12pt]{article}

\pdfoutput=1

\usepackage{amsthm,enumerate}
\usepackage{mathtools}
\usepackage[normalem]{ulem}
\usepackage[T1]{fontenc}
\usepackage[utf8]{inputenc}
\usepackage[thinc]{esdiff}
\usepackage{color}
\usepackage{fancyhdr}
\usepackage{calc}
\usepackage{url}
\usepackage{xcolor}
\usepackage{xspace}

\usepackage[all,cmtip]{xy} 
\usepackage{tikz-cd}

\usepackage[top=80pt,bottom=85pt,left=85pt,right=85pt]{geometry}
\usepackage{amsmath,amssymb,graphicx,float,color,tikz,cite,savesym,wasysym,subfigure,wrapfig}
\usepackage{caption}
\usepackage[debug,pageanchor=false]{hyperref}
\definecolor{link}{rgb}{1,0.45,0.05}
\hypersetup{colorlinks=true,linkcolor=link,citecolor=link,urlcolor=link,linktocpage}
 
\newcommand{\R}{\mathbb{R}}

\newcommand{\ii}{\mathrm{i}}
\newcommand{\dd}{\mathrm{d}}
\newcommand{\ee}{\mathrm{e}}

\newcommand{\N}{\mathbb{N}}

\newcommand{\mm}{\mathfrak m}
\newcommand{\di}{\mathsf{d}}

\newcommand{\sfd}{\mathsf d}

\newcommand{\Ch}{\mathsf{Ch}}

\theoremstyle{plain}
\newtheorem{lemma}{Lemma}[section]

\newtheorem{proposition}[lemma]{Proposition}
\newtheorem{corollary}[lemma]{Corollary}

\newtheorem*{theorem*}{Theorem}
\newtheorem*{maintheorem*}{Main Theorem}

\theoremstyle{definition}

\newtheorem{definition}[lemma]{Definition}
\newtheorem*{definition*}{Definition}
\newtheorem*{remark*}{Remark}
\newtheorem{remark}[lemma]{Remark}
\newtheorem{example}{Example}

\makeatletter
\@addtoreset{equation}{section}
\makeatother

\begin{document}

	\begin{titlepage}

	\begin{center}

	\vskip .5in 
	\noindent

	{\Large \bf{Harmonic functions and gravity localization}}

	\bigskip\medskip
	 G. Bruno De Luca,$^1$  Nicol\`o De Ponti,$^2$
	Andrea Mondino,$^3$	Alessandro Tomasiello$^{2,4}$\\

	\bigskip\medskip
	{\small 
$^1$ Stanford Institute for Theoretical Physics, Stanford University,\\
382 Via Pueblo Mall, Stanford, CA 94305, United States
\\	
	\vspace{.3cm}
	$^2$ 
	Dipartimento di Matematica e Applicazioni, Universit\`a di Milano--Bicocca, \\ Via Cozzi 55, 20126 Milano, Italy; 
\\	
	\vspace{.3cm}
	$^3$ Mathematical Institute, University of Oxford, Andrew-Wiles Building,\\ Woodstock Road, Oxford, OX2 6GG, UK
\\
	\vspace{.3cm}
$^4$ INFN, sezione di Milano--Bicocca
		}

   \vskip .5cm 
	{\small \tt gbdeluca@stanford.edu, nicolo.deponti@unimib.it,\\ andrea.mondino@maths.ox.ac.uk, alessandro.tomasiello@unimib.it}
	\vskip .9cm 
	     	{\bf Abstract }
	\vskip .1in
	\end{center}

	\noindent
	In models with extra dimensions, matter particles can be easily localized to a `brane world', but gravitational attraction tends to spread out in the extra dimensions unless they are small. Strong warping gradients can help localize gravity closer to the brane. 
	In this note we give a mathematically rigorous proof that the internal wave-function of the massless graviton is constant as an eigenfunction of the weighted Laplacian, and hence is a power of the warping as a bound state in an analogue Schr\"odinger potential. This holds even in presence of singularities induced by thin branes. 

	We also reassess the status of AdS vacuum solutions where the graviton is massive. We prove a bound on scale separation for such models, as an application of our recent results on KK masses. We also use them to estimate the scale at which gravity is localized, without having to compute the spectrum explicitly. For example, we point out that localization can be obtained at least up to the cosmological scale in string/M-theory solutions with infinite-volume Riemann surfaces; and in a known class of ${\mathcal N}=4$ models, when the number of NS5- and D5-branes is roughly equal.
		
	\noindent

	\vfill
	\eject

	\end{titlepage}
    
\tableofcontents

\section{Introduction} 
\label{sec:intro}

The existence of additional spacetime dimensions is a fascinating possibility that keeps attracting attention in theoretical physics. It is suggested by string theory, and would have dramatic consequences at very small length scales. 

For matter fields, there are two ways to avoid conflict with current observations. One is to assume that the extra dimensions describe a compact space $X_n$ of small size. Indeed, when spacetime is a direct product $M_4\times X_n$, several theorems give lower bounds on the lowest eigenvalue of the Laplace--Beltrami and other operators, guaranteeing that the Kaluza--Klein (KK) masses are large.\footnote{Here ``size'' can be understood both as $\mathrm{Vol}(X_n)^{1/n}$ or as $\mathrm{diam}(X_n)$, the largest distance between any two points. We will focus on the case where the external spacetime $M_4$ has four spacetime dimensions, but all our arguments are readily generalized.} In particular, if the masses of the spin-two KK fields are large, the usual $1/r^2$ behavior of gravity will only be modified at very small distances. The observed four-dimensional Planck mass $m_4$ is 
$$m_4= \sqrt{m_D^{D-2} \, \mathrm{Vol}(X_n)}\,,$$ 
where $m_D$ is the Planck mass of the $D$-dimensional gravity model. 

An alternative is to assume that matter particles are somehow stuck to a four-dimensional defect inside the higher-dimensional space. In this situation, their fields are ``localized'': they don't even depend on the extra dimensions, and there are no KK modes. There is however an exception: the metric field, which does depend on the extra dimensions and gives rise to KK modes. If the latter are light, they can modify Newton's law at large distances.

However, there is an interesting possible remedy to this problem, a way to ``localize'' gravity as well on a defect. This originates from the so-called Randall--Sundrum 2 (RS2) model \cite{randall-sundrum2}. It involves \emph{warped} products, spacetimes with a metric 
\begin{equation}\label{eq:vac}
	\dd s^2= \ee^{2A}(\dd s^2_{M_4} + \dd s^2_{X_n})\,,
\end{equation}
with the warping function $A$ depending on $X_n$. While in the direct product case $A=0$ and $X_n$ with infinite volume would lead to $m_4\to \infty$ and thus a non-dynamical graviton, this can be avoided with $A\neq 0$, where $m_4^2=m_D^{D-2}\int_{X_n} \sqrt{g_n} \ee^{(D-2)A}$, and it is sufficient to require that the integral is finite. The spin-two spectrum is given by the eigenvalues of a weighted Laplacian defined by $\Delta_A \xi = -\ee^{-(D-2)A} \nabla^m (\ee^{(D-2) A} \nabla_m\xi)$ \cite{bachas-estes,csaki-erlich-hollowood-shirman}. It is often studied by mapping it to a Schr\"odinger operator $(\Delta_0+ V)$, upon rescaling $\xi= s \psi$ by an appropriate function $s$, related to a power of $\ee^A$. The massless graviton corresponds to $\xi=\xi_0$, a constant; the corresponding $\psi_0$ is in $L^2(X_n)$, related to the finiteness of $m_4$. A peak in $\psi_0$ is interpreted intuitively as the graviton propagating preferentially around a defect of $X_n$, thus effectively keeping gravity four-dimensional. In the original RS2 model \cite{randall-sundrum2}, to be reviewed below in Section \ref{sub:mink}, one takes $D=5$, $X_n=\mathbb{R}$, and $A$ piecewise linear; $M_4$ is the Minkowski space, and the five-dimensional spacetime is obtained by gluing two pieces of AdS$_5$. The graviton wave-function $\psi_0$ is peaked near the origin. While the rest of the KK spectrum is continuous, the contributions of the massive spin-two fields are suppressed: their rescaled wave-functions $\psi_k$, $k>0$ are small near the origin of $\mathbb{R}$, which in turn can be seen from the peculiar shape of $V$, often called a ``volcano'' potential. 

The RS2 model has analogs with non-zero cosmological constant $\Lambda_4$: the Karch--Randall (KR) models \cite{karch-randall}. For $\Lambda_4>0$ the continuous part of the spectrum has a mass gap. For $\Lambda_4<0$, the integral $\int_{X_n} \sqrt{g_n} \ee^{(D-2)A}$ diverges. The lightest spin-two mass field is not massless, but is still much lighter than the rest of the KK tower. Moreover, the $\psi_{k>0}$ are still concentrated near the origin. These two effects combine to still give localization for small enough $|\Lambda_4|$.  We will review the RS2 and KR models in Section \ref{sec:rev}, using them to illustrate some general results that we will find useful later.

The first new result in this paper regards a further version of localization. On a smooth compact space, the zero mode of the (weighted) Laplace operator is easily seen to be constant by an integration-by-parts argument. In presence of singularities, this is not quite so obvious. For example, it was pointed out in \cite{crampton-pope-stelle} that the analogue potential $V$ is often $\sim\rho^{-2}$ near the origin; this can give rise to interesting bound states, depending on the self-adjoint extension one chooses \cite{essin-griffiths,derezinski-richard}. The four-dimensional Planck mass would then read
\begin{equation}\label{eq:m0-cps}
	m_4^2=m_D^{D-2}\int_{X_n} \dd x^n \sqrt{g_n} \ee^{(D-2)A} |\xi_0|^2\,.
\end{equation}
A normalizable $\xi_0$ would then make $m_4$ finite even if $\ee^{(D-2)A}$ is itself not normalizable.

Unfortunately, in Section \ref{sec:const} we will present rigorous mathematical arguments that go against this possibility. First we consider the definition of weighted Laplacian that is commonly considered in the theory of metric measure spaces, a class of possibly non-smooth spaces endowed with a reference distance and measure. Whenever it is linear, an assumption satisfied in the most relevant physical situations, this Laplacian arises via integrating by parts a suitable generalization of the Dirichlet energy (known in this framework as Cheeger energy),  and is automatically self-adjoint on its finiteness domain, leading to the usual spectral theorems familiar from the smooth case. Thus it is both natural from a mathematical point of view, and well suited for physics applications. In this context, we prove (Prop.~\ref{prop:ManSing}) that for spaces of interest in gravity compactifications the only eigenfunctions $\xi_0$ with zero eigenvalue are in fact the constant ones.

From this perspective, there is in a sense no need to select boundary conditions on the singularities; they are automatically selected by the metric-measure formalism. It is instructive to compare this with a point of view more similar to that advocated in \cite{crampton-pope-stelle}; namely, removing the singular locus $\mathcal{S}$ and working with the space of smooth functions with compact support on $X\setminus {\mathcal S}$. In Sec.~\ref{sub:polar} we find that this latter notion yields the former with Dirichlet boundary conditions on  ${\mathcal S}$, and the two are essentially equivalent (more precisely to obtain the former one needs to take the metric completion of the latter endowed with the $W^{1,2}$ norm), when the singular locus  $\mathcal{S}$ is a D$p$-brane for $p=1,\ldots, 7$. 

Finally, in Section \ref{sec:st} we consider models that realize localization in string theory, focusing on the AdS case. Using theorems we proved in \cite{deluca-deponti-mondino-t}, we first show that if the lightest spin-two $m_0 \neq 0$ is smaller than the cosmological constant, then the next mass $m_1$ cannot be arbitrarily large:
\begin{equation}\label{eq:intro-Sep}
	\frac{m_1^2}{|\Lambda_4|}<\frac{3528}{25}\left(1+ (D-2) \frac{(\mathrm{sup}|\dd A|)^2}{|\Lambda_4|}\right)\,.
\end{equation}
The norm of $\dd A$ is computed with respect to the $\dd s^2_{X_n}$ metric in (\ref{eq:vac}). (Generically one expects this last term to be small for solutions that are under control in supergravity, although there can be exceptions.)

More generally, the results in \cite{deluca-deponti-mondino-t} allow us to estimate the masses (and thus the extent of localization) without actually computing them. We apply this to two classes of models, which are the non-compact analogues of those we considered in \cite{deluca-deponti-mondino-t}. The first one consists in any compactification on Riemann surfaces, such as the usual Maldacena--Nu\~nez solution, where the internal space is a fibration over a Riemann surface of  infinite volume. The second class is ${\mathcal N}=4$-supersymmetric, and was worked out in this context by Bachas and Lavdas \cite{bachas-lavdas2}. In both cases one can achieve $m_0 \ll \Lambda_4$, while $m_1 \sim O(\Lambda_4)$. This means localization is achieved, but only for distances larger than the cosmological $L_4$. Unfortunately in these models one cannot make $m_1$ even larger, essentially because of (\ref{eq:intro-Sep}). We discuss how the aforementioned wave-function concentration mechanism might work for some modes; this would improve the situation and push the localization length scale lower. This goes out of the scope of the present work; it would be very interesting to pursue it further in the future.


\section{Effective models of gravity localization} 
\label{sec:rev}

In this section we will mostly review five-dimensional effective models that display gravity localization in various forms. We will end in Sec.~\ref{sub:pot} with some considerations about the analogue Schr\"odinger potential; this will serve as an introduction to the mathematical problem we will tackle in Sec.~\ref{sec:const}.

\subsection{Minkowski brane world} 
\label{sub:mink}

We begin with the most famous model of gravity localization, RS2 \cite{randall-sundrum2}. It consists of five-dimensional gravity with a cosmological constant and a four-dimensional defect: 
\begin{equation}\label{eq:rs2-action}
	m_5^{-3} S = \int \dd^5 x \sqrt{-g_5}\left( R_5 +\frac{12}{L_5^2}\right) - 4\lambda \int \dd^4 x \sqrt{-g_4}\,.
\end{equation}	

If the tension of the defect is tuned to $\lambda= 3/L_5$, a solution of the Einstein equations is
\begin{equation}\label{eq:rs2-metric}
	\dd s^2_5 = \ee^{2A} \dd s^2_{\mathrm{Mink}_4} + \dd r^2 \, ,\qquad A = -\frac{|r|}{L_5}\,.
\end{equation}
This is a warped product, as defined in the introduction. We recognize two pieces of AdS$_5$, glued together at $r=0$. 

As mentioned in the introduction, the (square) Planck mass $m_4^2=m_5^{3}\int \dd r\,\ee^{3A}$ in this case is finite. The spectrum of KK modes is obtained by analyzing the operator $\Delta_A = -\ee^{-2A} \partial_r (\ee^{4 A} \partial_r(\,\cdot\,) )$. Besides the expected $\mu=0$ eigenvalue, the continuous part of the spectrum is $\mathbb{R}_{>0}$, with no spectral gap  \cite{randall-sundrum2}.\footnote{We refer to section \ref{sec:spectral} for all the necessary terminology about spectral theory.} While this might appear discouraging for localization, the formal eigenfunctions associated with the continuous spectrum are small near $r=0$. This effect counteracts enough the absence of a spectral gap that the gravitational force localizes.

In general, the eigenvalue problem for a weighted Laplacian can be mapped to a Schr\"odinger problem as follows:
\begin{equation}\label{eq:schroedinger}
	-\ee^{-f}\nabla^m\left(\ee^f\nabla_m \xi\right) = \mu \xi\, ,\hspace{0.7cm} \xi=\ee^{-f/2}\psi
	\Leftrightarrow\ \left(\Delta_0 + U\right)\psi = \mu \psi \, ,\hspace{0.7cm} U= -\ee^{-f/2}\Delta_0 \ee^{f/2}\, . 
\end{equation}
To apply this to the spin-two operator and use the same conventions as in \cite{randall-sundrum2}, we first change coordinate such that $\dd r = \ee^{A}\dd z$, and (\ref{eq:rs2-metric}) becomes conformally flat, $\ee^{2A} (\dd s^2_{\mathrm{Mink}_4} + \dd z^2)$. Then taking $f=3A$ we obtain
\begin{equation}
	U= \frac{15}{4L_5^2} \left(1+\frac{|z|}{L_5}\right)^{-2}- \frac3{L_5} \delta(z)\,.
\end{equation}
The shape of this effective potential has earned it the moniker of \emph{volcano}: indeed it has a peak at the origin, and a negative delta that one can think of as a very thin and deep ``pipe''. The obvious $\mu=0$ eigenfunction $\xi=1$ becomes the single bound state $\psi_0=\ee^{3A/2}$, whose presence is allowed by the delta function. More importantly, the peak of $U$ gives an intuitive reason for the aforementioned suppression of the generalized eigenfunctions of the higher modes. Moreover, the presence of a continuous spectrum starting at $\mu=0$ can be explained by $U\to 0$ as $z\to \pm \infty$. 

In this case the spectrum is easy to analyze directly. In more complicated geometries with additional extra dimensions this is not always the case, and it is useful to have estimates for the KK masses, especially the low-lying ones. A general theory is available that provides bounds for the eigenvalues of a weighted Laplacian in terms of the internal diameter, weighted volume, or of the so-called \emph{Cheeger constants}. Informally, the latter quantify how much the internal space can be divided in pieces with small boundary and big bulk. When the weighted volume of the internal dimensions $\mathrm{Vol}_A(X):= \int \dd^n x \sqrt{g} \ee^{(D-2)A}$ (and hence the Planck mass) is finite, the first non-trivial Cheeger constant is
\begin{equation}\label{eq:h1}
	h_1 := \inf_{B} \frac{\text{Vol}_A(\partial B)}{\text{Vol}_A(B)} \, ,\qquad\text{Vol}_A( B)\leqslant \frac12 \mathrm{Vol}_A (X)\,.
\end{equation}
Here $\text{Vol}_A(\partial B) := \int_{\partial B} \sqrt{g|_{\partial B}}\ee^{(D-2)A}$, and $\text{Vol}_A(B) := \int_{B} \sqrt{g}\ee^{(D-2)A}$. As in the introduction, here the internal metric $\dd s^2_{X_n}$ is defined by $\dd s^2= \ee^{2A}(\dd s^2_{M_4} + \dd s^2_{X_n})$.

The origin of this bound lies in the variational approach to the eigenvalues; a $B$ is associated to a trial wavefunction with support over it. In general it is hard to compute $h_1$ exactly, as it involves minimization over infinitely many choices of $B$. But for the RS2 model we can consider $B=(r_0,\infty)$, and with $r_0>0$
\begin{equation}
	h_1 =\inf_{r_0} \frac{\ee^{3A(r_0)}}{\int_{r_0}^{\infty} \dd r\,\ee^{2A(r)}}= \inf_{r_0} \frac2{L_5}\ee^{-\frac{r_0}{L_5}}=0\,.
\end{equation}
Now \cite[Th.~4.2]{deluca-deponti-mondino-t} implies that the infimum of the continuous spectrum is zero.\footnote{That theorem assumes the internal space to have a property called $\mathrm{RCD}(K<0,\infty)$; this can be proven in a similar fashion as for D8-branes in string theory \cite[Th.~4.2]{deluca-deponti-mondino-t}.} 

The four-dimensional gravitational potential between two matter particles $M_1$, $M_2$ at $r=0$ is obtained from the two-point correlation function of two metric fluctuations $\delta g_{\mu \nu}$. Expanding the latter in spin-two KK modes with eigenfunctions $\psi_k(r)$ and masses $m_k$, one would obtain in general
\begin{equation}\label{eq:massive-V}
	V = \frac{G M_1 M_2}R \sum_{k=0}\ee^{-m_k R}\psi_k^2(0)\,
\end{equation}
where now $R$ is the distance in Mink$_4$.
In our case, the $m_0=0$ contribution provides the usual large-distance $d=4$ potential. The sum for $k>0$ is replaced by an integral over the continuous spectrum; an explicit analysis of the suppression near the origin gives \cite{randall-sundrum2} $\int_0^\infty \dd m \ee^{-m R} m/L_5^2$. This correction behaves as $\sim L_5^{-2} R^{-3}$ at small distances, so it is suppressed for large distances. It is in this sense that gravity localizes in this model.

It is natural to ask whether this model has a realization in string theory. The most natural analogue is a vacuum solution with two D3-brane stacks, as pointed out in \cite{verlinde-RS2,chan-paul-verlinde}. Indeed near such a stack the metric is asymptotic to the interior of AdS space, as in (\ref{eq:rs2-metric}) for $r\to \pm \infty$. (If one insists that the additional five dimensions should have the same topology at all values of $r$, finding models similar to (\ref{eq:rs2-metric}) becomes harder \cite{kallosh-linde-RS2,behrndt-cvetic-RS2,ceresole-dallagata-kallosh-vanproeyen}.) The holographic dual of the RS2 model is a CFT with a cut-off coupled to weakly gauged gravity \cite{witten-RS2-talk,gubser-RS2}.


\subsection{de Sitter} 
\label{sub:ds}

With the same action (\ref{eq:rs2-action}), if instead of fine tuning the brane tension as $\lambda= 3/L_5$ we take $\lambda>3/L_5$, we have the solution \cite{karch-randall}
\begin{equation}\label{eq:kr-ds-metric}
	\dd s^2_5 = \ee^{2A} \dd s^2_{\mathrm{dS}_4} + \dd r^2 \, ,\qquad \ee^A = \ee^{A_0} \sinh\frac{c-|r|}{L_5}\,,
\end{equation}
for $r\in (-c,c)$ and with $c$ defined by $\lambda = \frac3{L_5} \coth(c/L_5)$. 
The integration constant $A_0$ and the cosmological constant $\Lambda_4$ of dS$_4$ are redundant; we can fix this ambiguity by imposing $A(r=0)=0$, and $\Lambda_4=(L_5 \sinh(c/L_5))^{-2}$.

This KR model can be analyzed similar to the RS2 in Section \ref{sub:mink} \cite{karch-randall}. The squared Planck mass $m_4^2=m_5^{3}\int_{-c}^c \dd r\ee^{3A}$ is again finite.
The coordinate $z$ defined by $\dd r= \ee^{A}\dd z$ covers all $\mathbb{R}$. The effective potential
$U(z)$ again has a peak with a negative delta at the origin, but now its asymptotics is $\lim_{z\to \pm \infty} U(z)= \frac94 \Lambda_4$. Because of this, the continuous spectrum only starts at $\frac94 \Lambda_4$. There is of course again the bound state $\psi_0  = \ee^{3A/2}$, coming from the constant eigenfunction $\xi_0= 1$ of the weighted Laplacian.


\subsection{Anti-de Sitter} 
\label{sub:ads}

We will now consider models with $\Lambda_4<0$. Unlike for $\Lambda_4\geqslant 0$, here the massless graviton is absent. The lightest spin-two field can still be much lighter than the rest of the KK masses, and at a certain intermediate range its potential can still behave as $1/r$, as we will see. We consider here the $D=5$ KR model, and we will discuss string theory embeddings in Section \ref{sec:st}. 

The AdS version of the KR model is again obtained from (\ref{eq:rs2-action}), with $\lambda<3/L_5$ and solution \cite{karch-randall}
\begin{equation}\label{eq:kr-ads-metric}
	\dd s^2_5 = \ee^{2A} \dd s^2_{\mathrm{AdS}_4} + \dd r^2 \, ,\qquad \ee^A = \ee^{A_0} \cosh\frac{c-|r|}{L_5}\,,
\end{equation}
and $\lambda = \frac3{L_5} \tanh(c/L_5)$. Again we impose $A(r=0)=0$, with $\Lambda_4=-(L_5 \cosh(c/L_5))^{-2}$.

Since now the warping function diverges at infinity, the naive Planck mass $m_5^{3}\int\ee^{3A}$ is infinite, and so the usual $\xi=$ constant eigenfunction is not in $L^2$ (and as a consequence we cannot use $\psi= \ee^{3A/2}$). The lowest eigenvalue corresponds to a different eigenfunction $\xi_0$, and the analogue of the Planck mass for this light spin-two field is given now by (\ref{eq:m0-cps}).

Nevertheless, a version of localization is still at play in this model. An explicit analysis reveals that the lowest eigenvalue, while not zero, is much smaller than the higher ones: as $\Lambda_4\to 0$, one gets \cite{miemiec,schwartz-KR}\footnote{Analytically, the spectrum is given by the zeros of the function in \cite[(2.1)]{miemiec}. While this condition is still quite complicated to analyze, it can be written as a power series in $k^2\Lambda_4 L_5^2$ using the expansion of the hypergeometric function $F(a,b;c;z)$ around $z=1$ given e.g.~in \cite[(15.8.10)]{NIST}.}\footnote{It can be checked explicitly that for $|\Lambda_4|\to 0$ the odd eigenvalues become degenerate with the even ones. However, the odd eigenvalues correspond to odd eigenfunctions that as such vanish at the location of the brane and do not contribute to the 4d physics. }
\begin{equation}\label{eq:mass-rs-ads}
	m^2_0 \sim \frac32 \Lambda_4^2 L_5^2 \, ,\qquad m_{{2k >0}}^2\sim k(k+3)|\Lambda_4| +\frac12 k^3 \Lambda_4^{2} L_5^2 + O(k^5 |\Lambda_4|^3 L_5^4)\,.
\end{equation}

As in previous cases, we can check these results using Cheeger constants $h_k$. In a situation where $m_4$ is infinite, one considers the smallest of them $h_0$, which is defined similarly to (\ref{eq:h1}) but with a weaker constraint on the volume:
\begin{equation}\label{eq:h0}
	h_0 := \inf_{B} \frac{\text{Vol}_A(\partial B)}{\text{Vol}_A(B)}\,,\qquad \text{Vol}_A(B)<\infty\,.
\end{equation}
(This is automatically zero when $\mathrm{Vol}_A(X)$ is finite, as one sees by taking $B=X$.) Unlike in the Minkowski case, taking $B$ to be semi-infinite leads to an infinite $\mathrm{Vol}_A(B)$. A better result is obtained by considering a symmetric interval $(-r,r)$. In the limit $\Lambda_4\to 0$, $c\gg L_5$. Introducing $H(r)=\frac12 (\cosh^2x+1)$ and changing variable to $x=(c-r)/L_5$, with a rough approximation we get that the minimization is achieved for $r\sim c$, $x\to 0$:
\begin{equation}
	L_5^{-2}h_0 = \sqrt{-\Lambda_4}\inf_{x} \frac{\cosh^3x}{H(c/L_5)-H(x)}\sim \sqrt{-\Lambda_4 } H(c/L_5)^{-1} \sim 2(- \Lambda_4)^{3/2} \,.
\end{equation}

Applying \cite[Th.~4.2]{deluca-deponti-mondino-t} we obtain $ (\Lambda_4 L_5^2)^3 < m_0^2 L_5^2 < \frac{21}{10} \sqrt{3} (\Lambda_4 L_5^2)^2$, in agreement with (\ref{eq:mass-rs-ads}).\footnote{To apply \cite[Th.~4.2]{deluca-deponti-mondino-t} we also need $K$, a lower bound on the weighted Ricci curvature. This can be readily obtained directly from the equations of motion for \eqref{eq:rs2-action} as $K = 3 \Lambda_4$. } While of course the explicit result was already available, the present computation is a valuable warm-up for the ten-dimensional case of the next subsection.

The hierarchy in (\ref{eq:mass-rs-ads}) already indicates that a form of localization appears in this model.\footnote{The hierarchy between $m_0$ and $m_k$ can also be obtained by estimating the ratio between $h_0$ and the higher Cheeger constants $h_k$, as we will see in detail in Sec.~\ref{sec:st} in a more involved setup where analytic results for the spectra are not available. For the present case we obtain $h_0/h_1\sim |\Lambda_4|$ for $\Lambda_4\to0$, which agrees with the explicit result.} A four-dimensional observer testing gravity at distances $R$ with $ |\Lambda_4|^{-1/2} \ll R \ll (|\Lambda_4|L_5)^{-1}$ would not realize that the graviton has in fact a non-zero mass $m_0$, and would also not feel the effect of the $m_{k>0}$. However, the lower end of this length range is still cosmological, so this in itself would not be very satisfactory. We will now see that actually gravity remains four-dimensional well below this scale, thanks to the further effect of wave-function suppression near the origin, similar to that in the RS2 model.

To see this, consider again the gravitational potential. Even in AdS, for $R\ll |\Lambda_4|^{-1/2}$ we can still use the expression (\ref{eq:massive-V}). In the limit $\Lambda_4\to 0$, we can approximate the sum over $k$ as an integral, and use the estimates $\psi_0^2(0)\sim 2 +O(\sqrt{|\Lambda_4|}L_5)$, $\psi_k^2(0)\sim |\Lambda_4| L_5^2(\frac14 + \frac45 k)$ \cite{schwartz-KR}. This gives\footnote{\label{foot:quickV}The sub-leading $\sqrt{|\Lambda_4|} L^2  R^{-2}$ can also be estimated by noticing that the masses $m_k$ up to $k\sim 1/\sqrt{|\Lambda_4|}r$ contribute $O(1)$ to the sum in (\ref{eq:massive-V}), and the ones above it very little.}
\begin{equation}\label{eq:V-KR}
	V\sim G M_1 M_2 \left(\frac1R + \frac{\sqrt{|\Lambda_4|} L_5^2}{4 R^2} + \frac{L_5^2}{R^3}\left(\frac45 - \frac18|\Lambda_4|L_5^2\right)\right)\,.
\end{equation}
As $\Lambda_4\to 0$, the $R^{-1}+ L_5^2 R^{-3}$ behavior of the RS2 model is recovered. As in that model, the $R^{-3}$ term is negligible for $R \gg L_5$. The new $R^{-2}$ term is negligible if $R \gg \sqrt{|\Lambda_4|}L_5^2$, which is eventually weaker in the $\Lambda_4\to 0$ limit. So this model still displays localization: gravity would behave in a four-dimensional fashion at macroscopic distances.


\subsection{The Schr\"odinger potential} 
\label{sub:pot}

In the five-dimensional RS2 and KR models we reviewed in this section, the point of view of the Schr\"odinger potential in (\ref{eq:schroedinger}) was useful in developing intuition about the model's properties. However, we will now argue that in higher dimensions it can also be misleading in some respects. 

While on a smooth compact space it is easy to show that the zero mode of the Laplace operator is necessarily constant, in presence of singularities (and on a non-compact space) this might not necessarily be obvious. In particular, the point of view of the Schr\"odinger potential might suggest otherwise.

As a toy model, consider the radial part of the usual flat space Laplacian in $\mathbb{R}^d$, $\Delta_\mathrm{rad} \xi := - z^{1-d} \partial_z (z^{d-1} \partial_z \xi)$. The map (\ref{eq:schroedinger}) relates its spectrum to 
\begin{equation}\label{eq:schr-x2}
	- \partial^2_z \psi + U \psi := {\cal H} \psi = \lambda \psi  \, ,\qquad U= \frac{(d-1)(d-3)}{4 z^2}\,.
\end{equation}
The potential $1/z^2$ is of Calogero type, but here we need to assume that $z>0$, so we put an infinite barrier for $z\leqslant 0$ as in \cite{essin-griffiths}. For $d\neq 2$, $U\geqslant 0$ and one does not expect any bound states. The local solutions to (\ref{eq:schr-x2}) for $\lambda=0$ are $\psi=z^{(d-1)/2}$ and $z^{(3-d)/2}$, none of which are normalizable. (For $d=1$ and $d=3$, $U=0$ and these two local solutions become constant and linear.) These correspond to $\xi=1$ and $z^{2-d}$. Both are harmonic outside the origin, but the second is in fact the Green's function: it solves  $\Delta_\mathrm{rad} \xi=\delta$, up to an overall constant. 

The $d=2$ case in (\ref{eq:schr-x2}) deserves a separate treatment. The potential is now attractive:
\begin{equation}\label{eq:1/4z2}
	U=-\frac1{4z^2}\,.
\end{equation}
It is a priori possible to have bound states, as was discussed for example in \cite{essin-griffiths,derezinski-richard}. Intriguingly, the coefficient $1/4$ is a `critical' case in this study. The $\lambda=0$ solutions are $\psi=\sqrt{z}$ and $\sqrt{z}\log z$; again none of them are normalizable, and map respectively to $\xi=1$  and to the Green's function $\log z$ in $\mathbb{R}^2$. 

There is a subtlety, however. Recall that a rigorous definition of the Hamiltonian ${\mathcal H}$ also needs the data of a domain $D({\mathcal H})$ on which it is self-adjoint. The adjoint is defined of course by $\langle {\mathcal H}\psi_1,  \psi_2 \rangle= \langle \psi_1, {\mathcal H}^\dagger\psi_2\rangle$ and ${\mathcal H}$ is called Hermitian if $\langle {\mathcal H}\psi_1, \psi_2 \rangle= \langle \psi_1, {\mathcal H}\psi_2\rangle$, self-adjoint if moreover $D({\mathcal H})= D({\mathcal H^\dagger})$. For a more rigorous   introduction we refer to Sec.\;\ref{sec:spectral}. If one considers for example $D({\mathcal H})$ as the space of functions with compact support, usually $D({\mathcal H})\neq D({\mathcal H^\dagger})$. It might be possible, however, to \emph{extend} the domain by adding functions to it, such that ${\mathcal H}$ becomes self-adjoint.

Potentials proportional to $1/z^2$ admit a one-parameter choice $D_\lambda$ of self-adjoint extensions, and (\ref{eq:1/4z2}) in particular admits a single bound state $\psi= \sqrt{z}K_0(\sqrt{-\lambda}z)$. However, this corresponds to $\xi = K_0(\sqrt{-\lambda}z)$, which again behaves as $\sim \log z$ near the origin; thus, it solves  $\Delta_\mathrm{rad} \xi= \lambda \xi + \delta$ rather than $\Delta_\mathrm{rad} \xi= \lambda \xi$.\footnote{One might think at this point that one can try to define a self-adjoint extension of the Laplacian by adding the Green's function to the domain, working on a space where the support of the delta has been removed. We will analyze this possibility in Sec.~\ref{sub:polar}.}

A similar discussion is also relevant in string theory near D$p$-brane singularities. Writing the ten-dimensional metric as $\dd s^2_{10}= \ee^{2A}(\dd s^2_{M_4} + \dd s^2_{X_6})$, in Einstein frame we locally have $\ee^{2A}\sim H^{(p-7)/8}$ and $\dd s^2_{X_6}\sim \dd x^2_{p+1-d} + H (\dd z^2 + z^2 \dd s^2_{S^{8-p}})$, with $H$ a harmonic function of the transverse coordinates and $\dd x^2_{p+1-d}$ representing the metric parallel to the D$p$. We then have $\Delta_\mathrm{A,rad} \xi \sim -H^{-1} z^{p-8} \partial_z (z^{8-p} \partial_z \psi) $. The local discussion for $\lambda=0$ eigenvalues is then identical to the one above in flat $\mathbb{R}^d$, with $d=9-p$.

A more concrete example was discussed in \cite{crampton-pope-stelle}. Here spacetime is $\mathrm{Mink}_4\times X_6$, with $X_6$ non-compact; an NS5-brane stack fills spacetime and is smeared along an $S^2\subset X_6$, so that its back-reaction is characterized by a harmonic function $H$ of the two remaining directions. Symmetry reduces the warped Laplacian to an operator in the radial direction $\rho$ of this transverse $\mathbb{R}^2$, 
	$- \Delta_{A,\mathrm{rad}} \xi\propto \frac1{H \sinh 4 \rho} \partial_\rho ( \sinh2 \rho \partial_\rho \xi).$
Locally around $\rho=0$ the situation is similar to that discussed around (\ref{eq:1/4z2}). The one-parameter self-adjoint extension discussed there might raise hopes that a non-trivial bound state might exist. However, the $\lambda=0$ wave-function is $\xi= \log\tanh \rho$; a limit near $\rho=0$ shows that again $\Delta_{A,\mathrm{rad}} \xi= \delta$ rather than  $\Delta_{A,\mathrm{rad}} \xi= 0$.\footnote{A more recent analysis \cite{erickson-leung-stelle} shows that indeed without this mode the model does not display localization.}

In summary, the Schr\"odinger point of view might suggest that non-trivial self-adjoint extensions might give rise to non-trivial solutions $\psi_0$. But in practice we have seen that such solutions always map to $\xi_0$ which are Green's functions rather than genuine eigenfunctions of the weighted Laplacian. In the next section we will prove rigorously that this is always the case: the only zero mode in $L^2$ of the weighted Laplacian is the constant, even on spaces that are singular and non-compact. We will also comment on the difficulties that arise by including the Green function in the domain of the weighted Laplacian.



\section{Constant harmonic functions} 
\label{sec:const}

The aim of this section is to rigorously investigate the question introduced in Sec. \ref{sub:pot}. Namely:\footnote{In this section we will consider a suitable generalization of the weighted Laplacian, appropriate for the general metric-measure setting; thus, we will no longer stress the weight and we will drop the subscript ${}_A$.}
\\

 Let $X$ be a space with a well defined notion of Laplacian $\Delta$ and let $f$ be a global harmonic function, i.e.~let us suppose $\Delta f=0$ on $X$. Under which assumptions on the space can we infer that $f$ is constant?
\\

In particular, we will concentrate on metric-measure structures, which arise naturally in a vast number of situations and allow to describe relevant physical geometries, even in presence of singularities. We will find  that $f$ is forced to being constant in (R)CD$(K,\infty)$ spaces, as well as in all the other singular spaces arising from the backreaction of D$p$-branes and O$p$-planes in gravity solution.

To obtain the result, we first define in Sec.~\ref{sec: PEMS} the Laplacian and study some of its properties, and then in Sec.~\ref{sec: metric irreducible} we prove that $\Delta f = 0$ can only be solved by a constant $f$ in a certain class of spaces with the \emph{$L^2$-Liouville property}, which we show to include the physical spaces we are interested in. 
More specifically, in Sec.~\ref{sec: PEMS} we introduce the notion of \emph{Cheeger energy} of a function $f$ (denoted by $\Ch(f)$), as a generalization of $\int_X |\nabla f|^2$ to non-smooth spaces. Prop.~\ref{prop:equiv} shows that $\Ch(f)=0$ is equivalent to $\Delta f=0$; this can be thought of as the appropriate generalization of the usual integration by parts argument on compact non-singular spaces.
We then use the fact that a space with D-brane and/or O-plane singularities can be decomposed as a smooth weighted Riemannian manifold plus a singular set $S$. Owing to smallness (in the sense of Hausdorff codimension) of the singular set in these physical spaces, using Prop.~\ref{prop:ManSing} we can show that $\Ch(f) = 0$ implies that $f = \text{const.~} \mm\text{-a.e.}$ Putting these results together, we obtain
that the space satisfies the $L^2$-Liouville property, that is any zero mode of the weighted Laplacian is constant.

\subsection{Metric measure spaces}\label{sec: PEMS}
We start by introducing a very general class of metric measure structures where a notion of Laplacian ($\Delta$) is defined, see Def.~\ref{def:Laplacian}. We then conclude by analyzing its properties, which we will use in the next section to characterize the solution of the equation $\Delta f = 0$. Most of the material for the preliminary section is taken from \cite{AGS2}, to which we refer for all the details. 

We will deal with \emph{metric measure spaces}: they are triples $(X,\di,\mm)$ where $(X,\di)$ is a complete and separable metric space, and $\mm$ is a nonnegative, Borel and $\sigma$-finite measure.

We consider the following additional assumption that connects the distance $\di$ and the measure $\mm$:
\begin{equation}\label{eq: finite compact}
\forall K\subset X \ \textrm{compact}\ \exists\, r>0 \ : \ \mm(\{x\in X:\di(x,K)\leqslant r\})<\infty\, .
\end{equation}
Notice that \eqref{eq: finite compact} is satisfied whenever the measure $\mm$ is finite on bounded sets, and it is crucial in showing the existence of sufficiently many integrable Lipschitz functions. More precisely, assuming \eqref{eq: finite compact}, it is possible to prove that the class of bounded, Lipschitz functions $f\in L^2(X,\mm)$ with $|Df|\in L^2(X,\mm)$ is dense in $L^2(X,\mm)$, where $|Df|$ is the slope of the function $f$ defined as 
$$|Df|(x):=\limsup_{y\to x}\frac{|f(y)-f(x)|}{\di(y,x)}\, ,\qquad  \textrm{if} \ x \ \textrm{is an accumulation point},$$
and $|Df|(x):=0$ if $x$ is isolated.

A relaxed gradient of a function $f\in L^2(X,\mm)$ is a function $G\in L^2(X,\mm)$ for which there exist Lipschitz functions $f_n\in L^2(X,\mm)$ such that:
\begin{itemize}
\item $f_n\to f$ in $L^2(X,\mm)$ and $|Df_n| \rightharpoonup \tilde{G}$ in $L^2(X,\mm)$;\footnote{$\rightharpoonup$ denotes weak convergence. The expression ``$\mm$-a.e.'' means ``almost everywhere with respect to $\mm$''.}
\item $\tilde{G}\leqslant G$ $\mm$-a.e.~in $X$.
\end{itemize}

 It is possible to prove that the set of all the relaxed gradients of a function $f\in L^2(X,\mm)$ is a closed and convex subset of $L^2(X,\mm)$. Thus, when non-empty, there exists an element of minimal $L^2$-norm which is called \emph{minimal relaxed gradient} and denoted by $|Df|_{\ast}$. It is minimal also in the $\mm$-a.e.~sense, meaning that for any relaxed gradient $G$ of $f$ it holds $|Df|_{\ast}\leqslant G$ $\mm$-a.e. 

Given a function $f\in L^2(X,\mm)$,  the \emph{Cheeger energy} $\Ch(f)$ is defined (see \cite{Ch99, AGS2}) as
$$\Ch(f):=\frac{1}{2}\int_X |Df|^2_{\ast}\,\dd\mm\,,$$
with the convention $\Ch(f)=+\infty$ if $f$ has no relaxed gradients. As usual, we denote by $ D(\Ch)$ the domain of the Cheeger energy, i.e.~the set of $f\in L^2(X,\mm)$ with $\Ch(f)<\infty$.
It is possible to check \cite{AGS2} that 
\begin{equation}\label{eq:defW12norm}
\|f\|_{W^{1,2}}:=\left( \|f\|_{L^{2}}^{2} + 2 \Ch(f) \right)^{1/2}
\end{equation}
 defines a complete norm on the vector space  $ D(\Ch)$. The corresponding Banach space is denoted by $W^{1,2}(X,\sfd,\mm)$. When $(X,\sfd,\mm)$ is a smooth weighted Riemannian manifold, i.e.~$X$ is a smooth complete manifold with metric $g$ endowed with the geodesic distance and a weighted volume measure $\mm:=\ee^f\mathsf{\dd vol}_g$, then $W^{1,2}(X,\sfd,\mm)$ is the standard Sobolev space (which is a Hilbert space). However in the generality of metric measure spaces, $W^{1,2}(X,\sfd,\mm)$ is a priori a Banach (non-Hilbert) space, for instance this is the case when  $(X,\sfd,\mm)$ is $\mathbb{R}^{n}$ endowed with a non-Euclidean norm and the $n$-dimensional Lebesgue measure.

The Cheeger energy is clearly nonnegative and $\Ch(c)=0$ for any constant function $c\in L^2(X,\mm)$. Moreover it is a $2$-homogenous, convex and lower semicontinuous functional in $L^2(X,\mm)$. In smooth spaces and for smooth functions, $\Ch(f)$ reduces to the classical Dirichlet energy $\frac{1}{2}\int |\nabla f|^2$, but we will see below examples where the above definition makes sense in far more general cases.

In the next proposition we collect some useful properties of the minimal relaxed gradient. We refer to \cite{AGS2} for a proof.
\begin{proposition}\label{prop: prop mrg}
Let $f\in L^2(X,\mm)$ be a function admitting relaxed gradients. Then:
\begin{enumerate}
\item For any set $N\subset \R$ of null $\mathcal{L}^1$-measure, $|Df|_{\ast}=0$ $\mm$-a.e.~on the set $f^{-1}(N)$. 
\item  For any $g$ with $\Ch(g)<\infty$ and for any $c\in \R$ it holds $|Df|_{\ast}=|Dg|_{\ast}$ on the set $\{f-g=c\}$.
\item Suppose $\phi:J\subset\R\to\R$ is Lipschitz, and $J$ is an interval containing the image of $f$ (with $\phi(0)=0$ if $\mm(X)=\infty)$. Then $\phi(f)\in D(\Ch)$ and $|D\phi(f)|_{\ast}\leqslant |\phi'(f)||Df|_{\ast}$ $\mm$-a.e. 
\end{enumerate}
\end{proposition}

Now let us suppose \eqref{eq: finite compact}. As a consequence the set $D(\Ch)$ is dense in $L^2(X,\mm)$, and we can invoke the classical theory of gradient flows in Hilbert spaces to infer that for every $f\in L^2(X,\mm)$ there exists a unique locally Lipschitz curve $t\mapsto H_t(f)$ from $(0,\infty)$ to $L^2(X,\mm)$ such that 
$$\begin{cases}\frac{\mathrm{d}}{\mathrm{d} t}H_t(f)\in -\partial^{-}\Ch(H_t(f))\qquad \textrm{for a.e.}\ t \in (0,\infty),\\
H_t(f)\to f \qquad \textrm{as} \ t\to 0^{+}\, .
\end{cases}$$
Here $\partial^{-}\Ch\subset L^2(X,\mm)$ is the subdifferential of the functional $\Ch$, i.e.~given $f\in L^2(X,\mm)$ it holds $\ell\in\partial^{-}\Ch(f)$ if 
$$\int_X \ell(g-f)\dd\mm+\Ch(f)\leqslant \Ch(g) \qquad \textrm{for all} \ g\in L^2(X,\mm).$$

We refer to $H_t(f)$ as the \emph{heat flow} at time $t$ starting from $f$. Using the uniqueness of the curve $t\mapsto H_t(f)$, one can easily see that the heat flow satisfies the semigroup property $H_{t+s}=H_t\circ H_s$ for every $t,s>0$.

The heat flow has many regularizing effects. For instance, it is possible to prove that the right derivative $\frac{\mathrm{d}^{+}}{\mathrm{d} t}H_t(f)$ exists for any $t>0$ and it is equal to the element of minimal norm of $-\partial^{-}\Ch(H_t(f))$. This suggests to introduce the following: 
\begin{definition}\label{def:Laplacian}
	We write $f\in D(\Delta)$ if $f\in L^2(X,\mm)$ with $\partial^{-}\Ch(f)\neq \emptyset$; for $f\in D(\Delta)$ we denote by $\Delta f$ the element of minimal $L^2$-norm in $\partial^{-}\Ch(f)$ and we refer to it as the \emph{Laplacian} of $f$. 	
\end{definition}
Notice that we are assuming a natural integrability assumption on functions $f$ in the domain of the Laplacian $D(\Delta)$, namely by writing $\Delta f$ we are in particular assuming $f$ to be in $L^2(X,\mm)$. It is easy to check that the \emph{metric-measure} Laplacian that we have defined coincides with the weighted Laplacian whenever the underlying metric measure space is a smooth weighted Riemannian manifold.

An immediate consequence of the $2$-homogeneity of the Cheeger energy is that $H_t$ and $\Delta$ are $1$-homogeneous, i.e.~$H_t(c f)=c H_t(f)$ and $\Delta(c g)=c\Delta g$ for every $f\in L^2(X,\mm)$, $g\in D(\Ch)$ and $c\in\R$. However, $H_t$ and $\Delta$ are in general not additive, and thus not linear operators. Notice also that $H_t(c)\equiv c$ for every $t>0$ and every constant $c\in \R$ ($c=0$ if $\mm(X)=\infty$).

Regarding the Laplacian and the heat flow, still without assuming linearity, we have the following important properties (see {\cite[Prop.\;4.15 and Th.\;4.16]{AGS2}}.
\begin{proposition}\label{prop: prop heatflow} Let $(X,\di,\mm)$ be a metric measure space satisfying \eqref{eq: finite compact}. It holds:
\begin{enumerate}
\item For all $f\in D(\Delta)$ and $\phi:J\to \R$ Lipschitz, with $J$ a closed interval containing the image of $f$ (and $\phi(0)=0$ if $\mm(X)=\infty$), we have
\begin{equation}
\int_X \phi(f)\Delta f\, \dd\mm=\int_X \phi'(f)|Df|^2_{\ast}\, \dd\mm.
\end{equation}
\item Let $f\in L^2(X,\mm)$ and $e:\R\to [0,\infty]$ be a convex and lower semicontinuous function. Denoting by $E:L^2(X,\mm)\to [0,\infty]$ the functional defined by $E(f):=\int_X e(f)\,\dd\mm$, if $e'$ is locally Lipschitz in $\R$ with $E(f)<\infty$, then
\begin{equation}
E(H_tf)+\int_0^t\int_X e''(H_s f)|DH_s f|^2_{\ast}\,\dd\mm\,\dd s= E(f) \qquad \textrm{for every} \ t>0\, .
\end{equation} 
\end{enumerate}
\end{proposition} 

To summarize, we have given a general definition of a Laplacian in Def.~\ref{def:Laplacian} and we have shown some of its properties. However, for physical application we often also want the Laplacian to be a linear operator. This constraint excludes Finslerian structures\footnote{Recall that a Finsler space is a smooth manifold endowed with a distance induced by the length functional $\int \dd \lambda F(x,\partial_{\lambda} x)$, with $F(x, \partial_{\lambda} x)$ a norm in the velocity $\partial_{\lambda} x$, depending smoothly on the base point $x$. If for all $x$, the norm $F(x, \cdot)$ satisfies the parallelogram identity and thus it comes from a scalar product, one is back to the classical framework of Riemannian geometry. Hence, this is  a natural generalization of Riemannian geometry,  see for example \cite{chern-finsler} for a quick introduction.
Even though Riemannian structures are more common in physics, the language of Finsler geometry
is also useful in some contexts. For example, geodesics in stationary space-times are described by geodesics of a Finsler structure on appropriately  defined spatial slices \cite{gibbons-herdeiro-warnick-werner, gibbons-jacobi,chanda-gibbons-guha-maraner-werner}. See also \cite{lammerzahl-perlick} for a review of more speculative applications to physics.} and is achieved through the following property: the space $(X,\di,\mm)$ is called \emph{infinitesimally Hilbertian} if for any $f,g\in D(\Ch)$ it holds
\begin{equation}\label{eq: def inf hilb}
2\Ch(f)+2\Ch(g)=\Ch(f+g)+\Ch(f-g).
\end{equation}
This assumption ensures that the heat flow and the Laplacian are linear. In particular, $\Ch$ becomes a strongly local, symmetric Dirichlet form on $L^2(X,\mm)$, $H_t$ is the associated Markov semigroup and $\Delta$ its infinitesimal generator \cite{AGS1, G11}.
Recall that Dirichlet forms are particular quadratic forms that provide a way to generalize the Laplacian, a classic example being ${\cal E}(u):=\frac{1}{2}\int_M |\nabla u|^2\mathrm{dvol}_g$ on a Riemannian manifold $(M,g)$. See e.g.~\cite{fukushima-book} as a general reference on this topic.

\subsection{$L^2$-Liouville property of metric measure spaces}\label{sec: metric irreducible}
Having introduced the general setting we are working on and a suitable notion of Laplacian, we now characterize in Prop.~\ref{prop:equiv} the spaces where $\Delta f = 0$ implies that $f$ is a constant, saying that they satisfy the \emph{$L^2$-Liouville property}. In particular, one of the characterizations will show that $\Delta f = 0$ is equivalent to $\Ch(f)=0$. We then take advantage of the expression of the Cheeger energy outside the singular set to conclude (Prop.~\ref{prop:ManSing} and Rem.~\ref{rem: OpDp} below) that spaces with D-brane and O-planes backreactions satisfy this property. In doing this, it will be clear the advantage of working in our framework. We will also show that the $L^2$-Liouville property is satisfied in many other relevant classes of metric measure spaces.

Borrowing the terminology from the celebrated Euclidean result about the constancy of bounded harmonic functions, we start by introducing the following definition.
\begin{definition}\label{def:L2Liouv}
Let $(X,\di,\mm)$ be a metric measure space satisfying \eqref{eq: finite compact}. We say that $(X,\di,\mm)$ satisfies the $L^2$-Liouville property if for any function $f\in D(\Delta)$ with $\Delta f=0$ there exists $c\in \R$ such that $f=c$ $\mm$-a.e., i.e.~$\mm(\{x\in X \ : \ f(x)\neq c\})=0$. 
\end{definition}
We remark that we will not assume the infinitesimally Hilbertianity of $(X,\di,\mm)$.  On the one hand this allows higher generality in the spaces (for instance, Finsler manifolds enter the framework); on the other hand, the treatment is slightly more delicate as $\Delta$ is in general not linear and the standard spectral theory is not at disposal.  

Even if the definition \ref{def:L2Liouv} makes sense without imposing any condition on the metric measure space (beside \eqref{eq: finite compact}), we are essentially interested in situations where the support of the measure $\mm$ is a connected subset of $X$. In this case, our definition should be compared with the notion of \emph{irreducibility} of a Dirichlet form. Recall that a Dirichlet form is irreducible if the only invariant sets of the associated semigroup are negligible or co-negligible (i.e.~they are of measure zero or the complement has measure zero), where an invariant set $A\subset X$ is a measurable set such that $H_t(1_Af)=1_AH_tf$ $\mm$-a.e. for every $f\in L^2(X,\mm)$ and $t>0$ (see e.g.~\cite{fukushima-book}).

The next proposition is inspired by \cite[Proposition 2.3]{akr-jfa}. Characterization $(vi)$ is probably known to experts, at least for infinitesimally Hilbertian spaces where the Cheeger energy defines a Dirichlet form, but we remark that we did not find it explicitly stated in the literature.
\begin{proposition}\label{prop:equiv}
Let $(X,\di,\mm)$ be a metric measure space and assume \eqref{eq: finite compact}. The following are equivalent:
\begin{enumerate}[(i)]
\item $(X,\di,\mm)$ satisfies the $L^2$-Liouville property.
\item For any $f\in L^2(X,\mm)$ with $\Ch(f)=0$, there exists $c\in \R$ such that $f=c$ $\mm$-a.e. 
\item If $f\in L^2(X,\mm)$ admits minimal relaxed gradient that is equal $\mm$-a.e.~to the constant function $0$, then there exists $c\in \R$ such that $f=c$ $\mm$-a.e. 
\item For any $f\in L^2(X,\mm)$ there exists $c\in \R$ such that $\lim_{t\to \infty} H_tf=c$ in $L^2(X,\mm).$ 
\item If $f\in L^2(X,\mm)$ is such that $H_t f=f$ $\mm$-a.e.~for every $t>0$, then there exists $c\in \R$ such that $f= c$ $\mm$-a.e. 
\item If $f\in L^2(X,\mm)$ is such that $H_{t_0} f=f$ $\mm$-a.e.~for a certain $t_0>0$, then there exists $c\in \R$ such that $f=c$ $\mm$-a.e. 
\end{enumerate}
\end{proposition}
\begin{proof}
Notice that a function $f\in L^2(X,\mm)$ that satisfies the assumptions in $(ii)$ or $(iii)$ is in particular in the domain of the Laplacian with 
\begin{align*}
\Delta f=0 &\Longleftrightarrow 0\in \partial^{-}\Ch(f)\Longleftrightarrow \Ch(f)\leqslant \Ch(g) \ \forall g\in L^2(X,\mm)\\
&\Longleftrightarrow \Ch(f)=0 \Longleftrightarrow |Df|_{\ast}=0 \ \ \mm\text{-a.e.} 
\end{align*}
This proves the equivalence between $(i)$, $(ii)$ or $(iii)$. 

$(ii)\implies (iv)$: Since $\Ch$ is proper, lower semicontinuous, convex, with dense domain and $2$-homogeneous (thus even), and $H_t$ is defined as its gradient flow, we are in position to apply \cite[Theorem 5]{bruck-jfa} to infer that the strong $\lim_{t\to \infty}H_tf$ exists and is a minimum point of $\Ch$, and thus it is a function $f\in L^2(X,\mm)$ such that $\Ch(f)=0$. Using $(ii)$ the conclusion follows. 

$(iv)\implies (v)$: Let $f\in L^2(X,\mm)$ be such that $H_t f=f$ $\mm$-a.e.~for every $t>0$. Thus the strong $\lim_{t\to \infty}H_tf$ exists and it is equal to $f$. By $(iv)$ it follows $f=c$ $\mm$-a.e.~for some constant $c\in \R$.

$(v)\implies (ii)$: Let us suppose that $(v)$ holds and let $f\in L^2(X,\mm)$ be such that $\Ch(f)=0$. Then the curve $t\mapsto f$ from $(0,\infty)$ to $L^2(X,\mm)$ is Lipschitz continuous and a gradient flow of the Cheeger energy starting at $f$. By uniqueness we must have $H_t(f)=f$ $\mm$-a.e.~and $(ii)$ follows. 

It remains to show that $(vi)$ is equivalent to the previous points. It is obvious that $(vi)$ implies $(v)$. We now show that $(iii)$ implies $(vi)$.  Let $t_0>0$ and $f\in L^2(X,\mm)$ be such that $H_{t_0}f=f$. By applying Proposition \ref{prop: prop heatflow} with $e(x)=x^2$ we can infer that
$$E(H_{t_0}f)+2\int_0^{t_0}\int_X |DH_s f|^2_{\ast}\,\dd\mm\,\dd s= E(f)\, ,$$
and since $H_{t_0}f=f$ we have $\int_0^{t_0}\int_X |DH_s f|^2_{\ast}\,\dd\mm\,\dd s=0$. Thus for almost every $s\in [0,t_0]$ it holds $|DH_s f|_{\ast}=0$ $\mm$-a.e.~and using $(iii)$ it must be that $H_s f$ is $\mm$-a.e.~equal to a constant $c(s)$ for almost every $s\in [0,t_0]$. We can thus consider a sequence of time $s_n\to 0$ such that $H_{s_n}(f)=c(s_n)\to f$ where the convergence is $\mm$-a.e., and thus $f\equiv c$  for a certain constant $c\in \R$.
\end{proof}
We refer to \cite{akr-jfa} and \cite{Kaj17} for some other similar equivalent characterizations, at least in the context of Dirichlet forms, where the connection with the notion of irreducibility that we have recalled above is also discussed.

As we will see in the next examples, brought from \cite[Remark 4.12]{AGS2}, there exist spaces that do not satisfy the $L^2$-Liouville property.
\begin{example}
Consider the interval $X:=[0,1]\subset \R$. We endow it with the Euclidean distance $\di:=|\cdot|$ and a finite, Borel measure $\mm$ concentrated on $\mathbb{Q} \cap (0,1)$, i.e.~$\mm(X\setminus\{\mathbb{Q} \cap (0,1)\})=0$. It is clear that $(X,\di,\mm)$ is a metric measure space satisfying \eqref{eq: finite compact}. 

For any $n\in \N$, we consider an open set $A_n$ with $\mathcal{L}^1(A_n)\leqslant \frac{1}{n}$ and $\mathbb{Q} \cap (0,1)\subset A_n$, where $\mathcal{L}^1$ is the $1$-dimensional Lebesgue measure. To construct such a set $A_n$, one can simply consider an enumeration of the set $\mathbb{Q} \cap (0,1)=:\{e_k\}_{k\in \N}$ and define
$$A_n:=\bigcup_{k\in\N} B\left(e_k,\frac{1}{n2^{k}}\right),$$
where $B(e,r)$ is the open ball of center $e$ and radius $r$.
We then define the function $j_n:X\rightarrow \R$ as $j_n(x):=\mathcal{L}^1([0,x]\cap A_n^c)$, and from its expression we infer that $j_n$ is $1$-Lipschitz and $j_n(x)\rightarrow x$ strongly in $L^2(X,\mm)$. We consider now an $L$-Lipschitz function $f$ defined on $X$ and we set $f_n:=f\circ j_n$. Since $j_n$ is $1$-Lipschitz we have that $f_n$ is $L$-Lipschitz and converges strongly to $f$ in $L^2(X,\mm)$. Moreover, for every $n\in \N$
$$\int_{[0,1]}|Df_n|^2\,\dd\mm=0\,.$$
This last equality is a consequence of the fact that $\mm$ is concentrated on $\mathbb{Q} \cap (0,1)$ and $f_n$ is locally constant on $A_n$ (thus $|Df_n|(x)=0$ for every $x\in A_n$). It follows, by definition of the Cheeger energy, that $\Ch(f)=0$. To construct an example of space without the $L^2$-Liouville property it is thus sufficient to consider a measure $\mm$ such that there exists a Lipschitz function defined on $(X,\di)$ which is not equal to a constant $\mm$-a.e. For instance, one can fix an enumeration $\{e_k\}_{k\in \N}$ of the set $\mathbb{Q}\cap (0,1)$ and consider the Borel measure $\mm$ concentrated on $\mathbb{Q} \cap (0,1)$ and such that $\mm(\{e_k\})=1/2^k$, with the Lipschitz function $f(x)=x$. Notice that in this situation $\mathsf{supp}(\mm)=[0,1]$, by the density of $\mathbb{Q} \cap (0,1)$ in $[0,1]$. 

\end{example}

\begin{example}[Snowflake construction]
Let $(X,\di)$ be a complete metric space, and let $\alpha\in (0,1)$. Set $\di_{\alpha}:=\di^{\alpha}$ and consider the couple $(X,\di_{\alpha})$, which is still a complete metric space with the same induced topology of $(X,\di)$. (When $X=[0,1]$, $\di$ is the usual Euclidean distance, and $\alpha\in(0,1)$, there is a Lipschitz embedding of $(X,\di^\alpha)$ into a fractal, and in particular for $\alpha=\log 2/\log 3$ this fractal is the classic Koch snowflake \cite{assouad}.)
It is easy to show that every $\di_{\alpha}$-absolutely continuous curve is constant and thus, using the characterization of Sobolev class via test plans, it follows that for every Borel measure $\mm$ on $(X,\di)$ and every $f\in L^2(X,\mm)$ it holds $\Ch(f)=0$ (see e.g.~\cite[Exercise 2.1.14]{gigli-pasqualetto-book}). Thus, if there exists a measure $\mm$ and functions $f\in L^2(X,\mm)$ that are not $\mm$-a.e.~equal to a constant, one can use this construction to produce examples of spaces without the $L^2$-Liouville property. In particular, every complete Riemannian manifold endowed with the power $\alpha\in(0,1)$ of the geodesic distance and the volume measure does not satisfy the $L^2$-Liouville property. 
\end{example}

As we will see in the next proposition, the class of spaces satisfying the $L^2$-Liouville property is sufficiently large to contain all the spaces with a synthetic lower bound on the Ricci curvature and no upper bound on the dimension, the $\mathsf{CD}(K,\infty)$ spaces for short.\footnote{\label{foot:CD}Roughly speaking, these are spaces with certain singularities, on which a generalization of a lower bound on the Ricci curvature still makes sense. We refer to our earlier \cite{deluca-deponti-mondino-t,deluca-deponti-mondino-t-entropy} for informal discussions of these spaces, and to \cite{LoVi,St} for more mathematical details.}
\begin{proposition}\label{prop:CDIrr}
Let $(X,\di,\mm)$ be a $\mathsf{CD}(K,\infty)$ space for some $K\in \R$. Then $(X,\di,\mm)$ satisfies the $L^2$-Liouville property.
\end{proposition}
\begin{proof}
Let $f\in L^2(X,\mm)$ with $\Ch(f)=0$. By Proposition \ref{prop:equiv} it is sufficient to show that there exists $c\in \R$ such that $f=c$ $\mm$-a.e. Let $\{f_n\}$ be a sequence of bounded Lipschitz functions, $f_n\in L^2(X,\mm)$, such that $f_n\to f$ and $|Df_n|\to |Df|_{\ast}$ in $L^2(X,\mm)$. The existence of such a sequence in ensured by \cite[Lemma 4.3]{AGS2}. By applying the weak local Poincaré inequality established in \cite[Theorem 1]{rajala-calcvar} to the sequence $\{f_n\}$, and taking the limit as $n\to \infty$, we can infer that for every $x\in X$ and $r>0$ 
$$\int_{B(x,r)}|f-\langle f\rangle_{B(x,r)}|\,\dd\mm=0\, ,$$
where $\langle f\rangle_{B(x,r)}$ denotes the mean of the function $f$ in the ball $B(x,r)$. In particular $f(y)=\langle f\rangle_{B(x,r)}$ for $\mm$-a.e.~$y\in B(x,r)$. Since $r$ is arbitrary, this gives that $f$ is $\mm$-a.e.~equals to a constant. 
\end{proof}
\begin{remark}
The $L^2$-Liouville property for $\mathsf{CD}(K,\infty)$ spaces seems to be not explicitly stated in the literature. Notice that in these spaces the support of the measure is a connected subset of $X$ (actually a length space), by the very definition of the $\mathsf{CD}(K,\infty)$ condition. The subclass of $\mathsf{CD}(K,\infty)$ spaces which are also infinitesimally Hilbertian is known as the class of $\mathsf{RCD}(K,\infty)$ spaces. In this case the irreducibility of the Cheeger energy  was explicitly noticed in \cite{DelloSchiavo-Suzuki} (where also the more general, but still infinitesimally Hilbertian, $\mathsf{RQCD}$ spaces were considered). The proof we have given here follows a different strategy from the one applied in \cite{DelloSchiavo-Suzuki}.
\end{remark} 

Another class of non-smooth spaces which has been widely studied in recent years is the one of PI spaces \cite{heinonen-koskela-acta}. These are metric measure spaces $(X,\di,\mm)$ which satisfy a local doubling inequality and a weak Poincar\'e inequality, but no curvature bound a priori. Also such spaces satisfy the $L^2$-Liouville property: 

\begin{proposition}
Let $(X,\di,\mm)$ be a PI space. Then $(X,\di,\mm)$ satisfies the $L^2$-Liouville property.
\end{proposition}

\begin{proof}
By definition, a PI space satisfies a weak local Poincar\'e inequality. Moreover, for every $f\in L^2(X,\mm)$ with $\Ch(f)=0$, there exists a sequence $\{f_n\}$ of bounded Lipschitz functions, $f_n\in L^2(X,\mm)$, such that  $f_n\to f$ and $|Df_n|\to |Df|_{\ast}$ in $L^2(X,\mm)$ (see for instance \cite[Corollary 5.15]{BjornBjorn}). We thus have all the ingredients to follow verbatim the proof of Proposition \ref{prop:CDIrr}.
\end{proof}

Finally, we can also show that the physical backreaction of D-branes and O-planes gives rise to a singular but reducible space, thanks to the following:
\begin{proposition}\label{prop:ManSing}
Let  $(X,\di,\mm)$ be a metric measure space with the following properties:
\begin{itemize}
\item there exists a closed set $S\subset X$ such that  $(X\setminus S,\di\llcorner_{X\setminus S} ,\mm\llcorner_{X\setminus S})$ is isomorphic  to an open subset $O$ of a smooth weighted Riemannian manifold $(M,g, e^\phi \mathrm{dvol}_g)$, meaning that there exists an isometry  $\psi$ from $(X\setminus S,\di\llcorner_{X\setminus S})$ to $(O,\di_g\llcorner_{O})$ which sends the measure $\mm\llcorner_{X\setminus S}$ to the weighted volume measure $e^\phi \mathrm{dvol}_g \llcorner_{O}$; 
\item $\mm(S)=0$;
\item $X\setminus S$ is connected by rectifiable curves, i.e.~for every $x,y\in X\setminus S$ there exists a curve $\gamma:[0,1]\to X\setminus S$ of finite length such that $\gamma(0)=x$ and $\gamma(1)=y$. 
\end{itemize}
Then $(X,\di,\mm)$ satisfies the $L^2$-Liouville property.
\end{proposition}

\begin{proof}
Let $f\in L^2(X,\mm)$ with $\Ch(f)=0$. The result follows if we show that $f$ is constant $\mm$-a.e., thanks to Proposition \ref{prop:equiv}. 
\\Since by assumption $(X\setminus S,\di\llcorner_{X\setminus S} ,\mm\llcorner_{X\setminus S})$ is isomorphic to an open subset of a smooth weighted Riemannian manifold, by the expression of the Cheeger energy we infer $|\nabla f|=0$ on $X\setminus S$ $\mm$-a.e., where $\nabla f$ is the standard gradient in smooth Riemannian geometry.
\\Let $x,y\in X\setminus S$ be arbitrary points. Then, by assumption, there exists  a curve $\gamma:[0,1]\to X\setminus S$ of finite length such that $\gamma(0)=x$ and $\gamma(1)=y$.  
\\Then, by the fundamental theorem of calculus,
$$
|f(x)-f(y)|=\left| \int_{0}^{1} \nabla f(\gamma(t)) \cdot \dot{\gamma}(t)\, dt  \right|=0\, ,
$$
where $\dot{\gamma}$ is the velocity of the curve (which is well defined for a.e.~$t\in [0,1]$ by the rectifiability of $\gamma$) and $\cdot$ denotes the Riemannian scalar product on $X\setminus S$ which we can think as identified to an open subset of a smooth weighted Riemannian manifold.
\\It follows that there exists a constant $c\in \R$ such that $f=c$ on $X\setminus S$. Since by assumption $\mm(S)=0$, we conclude that $f(x)=c$ for  $\mm$-a.e.~$x\in X$.
\end{proof}

\begin{remark}\label{rem: OpDp}
The assumptions of Proposition \ref{prop:ManSing} are satisfied in the physically relevant situation of a smooth weighted manifold outside of a singular set where the metric-measure structure is asymptotic to localised sources of co-dimension at least 2, such as O-planes or D$p$-branes of co-dimension $\geqslant 2$. Thus such metric measure spaces satisfy the $L^2$-Liouville property. Also D$8$-branes and O$8$-planes satisfy the $L^2$-Liouville property, but they require a separate discussion. Recall that, in a neighborhood $\{|r|<\epsilon\}$ of the closed singular set  $\{r=0\}$,  the metric is of the form
        \begin{equation}\label{eq:metricEndAs8}
	g= \dd x^2_{9-d}  + (1- h_8 |r|) \dd r^2 
	\end{equation}
	where,   $h_8>0$  is a positive constant for D$8$ (resp. $h_8<0$ for O$8$), and the measure  is given by 
	$$
	\mm\llcorner_{\{|r|<\epsilon \}}= \sqrt{1- h_8 |r|} \, \mathsf{\dd vol}_g \llcorner_{\{|r|<\epsilon \}}
	$$
	where $ \mathsf{\dd vol}_g$ is the Riemannian volume measure associated to $g$.
	By applying Proposition \ref{prop:ManSing} to the smooth part of space, we obtain that if $f$ is an $L^2$ harmonic function, then there exist two constants $c_1,c_2$ such that $f|_{r<0}=c_1$ and $f|_{r>0}=c_2$. We claim that it must hold $c_1=c_2$.  Assume by contradiction $c_1\neq c_2$. Then $f$ would have a jump along the singular set $\{r=0\}$. However the metric and the measure are bounded (above and below away from 0) in $\{|r|<\epsilon \}$; thus such an $f$ would not be an element of $W^{1,2}$, yielding a contradiction (recall that a $W^{1,2}$-function  cannot jump along a set of co-dimension one $S$, provided the ambient metric-measure structure is not too degenerate near $S$). 
\end{remark}
\subsection{Spectral theory}\label{sec:spectral}
In this section we review some terminology and basic definitions of spectral theory. The material is classical and can be found for instance in \cite{Reed-Simon-1,davies-book}. For a gentle introduction see also T. Tao's blog \cite{terryTspect}. In the final part of the section we then notice how the metric measure Laplacian enters in this framework.

We start with a linear operator on a complex Hilbert space $(H, \langle \cdot,\cdot \rangle)$. By this we mean a couple $(A,D(A))$ where $D(A)$ is a dense subset of $H$, called the domain of $A$, and $A:D(A)\rightarrow H$ is linear. We remark that we work with \emph{unbounded} operators, meaning that $D(A)$ can be strictly contained in $H$ (and this is the typical situation we will encounter). The operator $A$ is symmetric if $\langle Ax,y\rangle=\langle x,Ay\rangle$ for every $x,y\in D(A)$, and nonnegative if $\langle Ax,x\rangle$ is a nonnegative real number for every $x\in D(A)$. We say that $A$ is closed if the set $\{(x,Ax) :  x\in D(A)\}$ is closed as a subset of $H\times H$.
 
The \emph{adjoint} of the operator $A$ is the couple $(A^\dagger,D(A^\dagger))$  where $D(A^\dagger)$ is the set of vector $y\in H$ such that the map $x\mapsto \langle Ax,y\rangle$ is a bounded linear operator on $D(A)$. For such $y$, we define $A^\dagger y$ has the only element of $H$ such that $\langle Ax,y\rangle=\langle x,A^\dagger y\rangle$ for every $x\in D(A)$, $y\in D(A^\dagger)$. Notice that the well posedness of this definition comes from the density of $D(A)$ in $H$ and from an application of the Riesz representation theorem. One can easily see that $A^\dagger$ is a linear operator and, when $A$ is symmetric, $A^\dagger$ is an extension of $A$, i.e. $D(A)\subset D(A^\dagger)$ and $A=A^\dagger$ on $D(A)$. In general $D(A)$ can be strictly contained in $D(A^\dagger)$, and we call \emph{self-adjoint} the symmetric operators $A$ such that $D(A)=D(A^\dagger)$. The subclass of self-adjoint operators  is of great importance,  as the spectral theorem applies to those (see \cite{davies-book}).
\\When working with operators of differential nature,  usually the initial domain where the operator $A$ is defined is ``small'' (think of a differential operator initially defined only on smooth functions), leading to a ``large'' $D(A^\dagger)$ and thus to the lack of self-adjointness of $A$. One is thus interested in extending $A$ by enlarging the initial domain in order to obtain a self-adjoint operator (typically one passes from smooth functions to a suitable Sobolev space). It is possible that an operator admits many self-adjoint extensions, and we call \emph{essentially self-adjoint} the important class of operators that admit a unique self-adjoint extension.

The \emph{regular values} $\rho(A)$ of an operator $A$ are the values $\lambda\in \mathbb{C}$ such that $(\lambda {\rm Id}-A)$ has a bounded inverse. The \emph{spectrum} $\sigma(A)$ is the set of numbers $\lambda$ that are \emph{not} regular values. A non-zero function $f\in D(A)$ is an \emph{eigenfunction} of $A$ of \emph{eigenvalue} $\lambda$ if  $A f=\lambda f$. Notice that for a nonnegative operator $A$, all the eigenvalues lie in the set $[0,\infty)$. The set of all eigenvalues constitutes the \emph{point spectrum} while the \emph{discrete spectrum} $\sigma_d(A)$ is the set of eigenvalues which are isolated in the point spectrum and with finite dimensional eigenspace. Finally the \emph{essential spectrum} is defined as $\sigma_{ess}(A):=\sigma(A)\setminus \sigma_{d}(A)$.

The definitions that we have introduced in this section are of particular interest for us since they can be applied to the Laplacian defined in Def.\;\ref{def:Laplacian} if the underlying metric measure space is infinitesimally Hilbertian. In this case $\Delta$ is a densely defined, nonnegative, self-adjoint operator on its domain $D(\Delta)\subset L^2$. We can thus study it by taking advantage of the important results of spectral theory.

We have in particular that eigenfunctions of the Laplacian relative to different eigenvalues are orthogonal. For spaces of finite measure, constant functions are eigenfunctions relative to $\lambda_0=0$, hence any other eigenfunction has null mean value.

We can also make use of the variational characterization of the eigenvalues. More precisely, given $f\in W^{1,2}(X, \di, \mm)$, $f\not\equiv 0$, we introduce the Rayleigh quotient defined as
\begin{equation}\label{eq: Rayleigh} 
\mathcal{R}(f):= \frac{2\Ch(f)}{\int_X |f|^2\,\dd\mm}\, .
\end{equation}
Notice that for any eigenfunction $f_{\lambda}$ of eigenvalue $\lambda$, it holds $\lambda=\mathcal{R}(f_{\lambda}).$ We can then infer that the set of eigenvalues below $\inf \sigma_{ess}(\Delta)$ is at most countable and, listing the elements in increasing order $\lambda_0<\lambda_1\leqslant \ldots\leqslant \lambda_k\leqslant \ldots$, the following characterization holds
\begin{equation}\label{eq:defeigk}
\lambda_k= \min_{V_{k+1}}\, \max_{f\in V_{k+1},\\ f\not\equiv 0} \ \mathcal{R}(f)\, ,
\end{equation}
where $V_{k}$ denotes a $k$-dimensional subspace of $W^{1,2}(X,\di,\mm)$.

\subsection{The singular set of D$p$-branes is polar}
\label{sub:polar}

The D$p$-branes are examples of singular spaces (more precisely they can be modelled by possibly non-smooth metric measure spaces), which are smooth weighted Riemannian manifolds outside of a singular set $\mathcal{S}$. 

As we saw earlier, it is interesting to study the spectrum of the Laplacian on such spaces. In the previous sections, we recalled the definition of Laplacian for metric measure spaces (in terms of the Cheeger energy) and how it is linked to standard spectral theory; a natural way to address the problem is thus to study the spectrum of such Laplacian. An a-priori different approach would be to restrict the Laplace operator to smooth functions with compact support outside of the singular set and study its spectral properties.  The goal of this section is to show that these two approaches are equivalent for D$p$-branes, $p<8$ (see Corollary \ref{Cor:Delta=Delta0} for the precise statement).

First, let us define the metric measure spaces we will consider. We refer to our previous works \cite{deluca-deponti-mondino-t, deluca-deponti-mondino-t-entropy} for discussions and further references. 

\begin{definition}[Asymptotically D$p$-brane  metric measure spaces]\label{def: Dp-brane mms}
An \emph{ asymptotically D$p$-brane  metric measure space} is a smooth Riemannian manifold $(X,g)$ that is glued (in a smooth way) to a finite number of ends where the metric $g$ is asymptotic to a D$p$-brane singularity in the following sense. 
\begin{itemize}
\item Case $p=0,1,\ldots, 6$.  In the end, as $r\to 0$, the metric is asymptotic to
	\begin{equation}\label{eq:metricEndAs1}
	\dd x^2_{p+1-d} +\left( \frac{r_0}{r} \right)^{7-p} \left( \dd r^2+ r^2 \dd s^2_{\mathbb{S}^{8-p}} \right) 
	\end{equation}
	with $r_0^{7-p}= g_s (2\pi l_s)^{7-p}/((7-p) \mathrm{Vol}(\mathbb{S}^{8-p}))$;  as usual $g_s$ is the string coupling and $l_s$ is the string length.
        \item Case $p=7$. In a neighborhood $\{r<\epsilon\}$ of the closed singular set  $\{r=0\}$, as $r\to 0$,  the metric is asymptotic to
        \begin{equation}\label{eq:metricEndAs7}
	\dd x^2_{8-d}  - \frac{2\pi}{g_s} \log (r/r_0)  \left( \dd r^2+ r^2 \dd s^2_{\mathbb{S}^{1}} \right).
	\end{equation}
\end{itemize}
In all the above cases, we endow $X$ with a weighted measure, and view it as a metric measure space $(X,\di,\mm)$ where:
\begin{itemize}
\item The distance $\di$ between two points $x,y \in X$ is given by 
$$\di(x,y):=\inf_{\gamma\in \Gamma(x,y)} \int \sqrt{g\left(\gamma'(t),\gamma'(t)\right)} \dd t\,,$$
where $\Gamma(x,y)$ denotes the set of absolutely continuous curves joining $x$ to $y$.
\item The measure $\mm$ is a weighted volume measure $\mm:=\ee^f\mathsf{\dd vol}_g$, with the function $\ee^f$ smooth outside the tips of the ends and gives zero mass to the singular set.
Near the singularity, the weight has the following asymptotics:
\begin{equation}\label{eq:f-Dbranes}
	\ee^f \sim H^\frac{p-7}{2}\qquad\qquad \text{for}\; r\to 0 \;,
\end{equation}
and, near the singularity,  
\begin{equation}\label{eq:Harm}
	H\sim \left\{\begin{array}{llr}
		   (r/r_0)^{p-7}  &&0 \leqslant p < 7 \\
		  - \frac{2\pi}{g_s}\log (r/r_0) & &p = 7 
	\end{array}	\right.\qquad\qquad \text{for}\; r\to 0 \;,
\end{equation}
where $r_0^{7-p}= g_s (2\pi l_s)^{7-p}/((7-p) \mathrm{Vol}(\mathbb{S}^{8-p}))$ for $p<7$.
\end{itemize}
\end{definition}

Before, see \eqref{eq:defW12norm},  we recalled the notion of Sobolev space  $W^{1,2}(X,\sfd,\mm)$ associated to a  metric measure space $(X, \sfd, \mm)$.
Notice that, as we proved in \cite[Proposition 6.4]{deluca-deponti-mondino-t-entropy}, this Sobolev space is Hilbert for asymptotically D$p$-brane  metric measure spaces. 

Let us now recall the notion of \emph{polar} set in $X$.  The rough idea is that, as sets of zero $\mm$-measure are ``invisible''  by  Lebesgue $L^{2}$-functions (or, more precisely, two Borel functions which agree outside of a set of measure zero correspond to the same element in the Lebesgue space  $L^{2}(X,\mm)$), polar sets are ``invisible''  by Sobolev $W^{1,2}$-functions (or, more precisely, two Borel functions which agree outside of a polar set correspond to the same element in the Sobolev space  $W^{1,2}(X,\sfd,\mm)$).

\begin{definition}[Polar set]
Let $(X, \sfd, \mm)$ be a metric measure space. A Borel subset $E\subset X$ is said to be \emph{polar} if
\begin{equation}\label{eq:DefPolar}
\inf\{ \|u\|_{W^{1,2}} \colon u \text{ Lipschitz, }  u=1 \text{ on a neighbourhood of } E \text{ and }  0\leqslant u\leqslant 1\}=0.
\end{equation}
\end{definition}
Equivalently, a subset $E\subset X$ is polar if it has zero 2-capacity (where the 2-capacity of $E$ is defined as the left hand side of \eqref{eq:DefPolar}).

\begin{proposition}\label{Prop:SingPolar}
Let  $(X, \sfd, \mm)$ be an asymptotically D$p$-brane metric measure space in the sense of Definition \ref{def: Dp-brane mms}. In case $p=6,7$, assume also that, for each end, the Riemannian factor   $\dd x^2_{p+1-d}$ has finite volume. Denote by  $\mathcal{S}$ the (minimal, under inclusion) closed singular set such that $X\setminus \mathcal{S}$ is isomorphic to an open subset of a smooth weighted Riemannian manifold. 

Then $\mathcal{S}$ is polar.
\end{proposition}

\begin{proof}
\textbf{Case $p=0,\ldots, 5$}. The statement is trivially true, as $\mathcal{S}=\emptyset$: indeed, the singularity is at infinity and $X$ is a smooth weighted Riemannian manifold.
\smallskip

For the cases $p=6,7$, notice that it is enough the prove that, for each end, the singular set $\{r=0\}$ is polar. 

\textbf{Case $p=6$}. Consider the coordinates  $(x, \Theta, r)\in M^{7-d}\times \mathbb{S}^{2} \times [0,\infty)$ as in  \eqref{eq:metricEndAs1}.
After the change of variable $\rho=2 \sqrt{r}$, in a neighbourhood of $\{\rho=0\}$, the metric is asymptotic to 
$$ 
\dd x_{7-d}^{2}+r_{0} \dd\rho^{2}+\, \frac{r_{0}}{4}\rho^{2} \dd s^{2}_{\mathbb{S}^{2}}\, ,
$$
with measure asymptotic to (up to a multiplicative constant)
$$
\rho^{3} {\rm vol}_{M^{7-d}}(\dd x)\, {\rm vol}_{\R}(\dd \rho)\, {\rm vol}_{\mathbb{S}^{2}} (\dd \Theta)\, .
$$
For each $k>0$, consider the following Lipschitz functions:
\begin{equation}\label{eq:defpsi1k}
\psi_{k}(\rho):=
\begin{cases}
1 \; \text{ for } \rho \leqslant k^{-1} \\
2- k \rho \; \text{ for } \rho \in  [k^{-1}, 2k^{-1}] \\
0 \; \text{ for } \rho \geqslant  2k^{-1}.
\end{cases}
\end{equation}
Using that the Riemannian factor $\dd x_{7-d}^{2}$ has finite volume, it is a straightforward computation to check that $\|\psi_{k}\|_{W^{1,2}}\to 0$ as $k\to \infty$. Thus, for $E=\{\rho=0\}$, the infimum in \eqref{eq:DefPolar} is zero and the set $\{\rho=0\}$ is polar.
\smallskip

\textbf{Case $p=7$}. Consider the coordinates  $(x, \Theta, r)\in  M^{8-d}\times \mathbb{S}^{1} \times [0,\infty)$ as in \eqref{eq:metricEndAs7}.
After the change of variable $\rho=\int_{0}^{r} \sqrt{\log(s)} \dd s$, in a neighbourhood of $\{\rho=0\}$, the metric is asymptotic to 
$$
\dd x_{8-d}^{2}+\frac{2\pi}{g_{s}} \left(\dd \rho^{2}+f(\rho) \dd s^{2}_{\mathbb{S}^{1}} \right),
$$
with measure asymptotic to (up to a multiplicative constant)
$$
\sqrt{f(\rho)}\,   {\rm vol}_{M^{8-p}}(\dd x)\, {\rm vol}_{\R}(\dd \rho)\, {\rm vol}_{\mathbb{S}^{1}} (\dd \Theta),
$$
where $f(\cdot)$ satisfies
$$
0\leqslant f(\rho)\leqslant \rho^{2}.
$$
Let $\psi_{k}$ be defined as in \eqref{eq:defpsi1k}. Using that Riemannian factor $\dd x_{8-d}^{2}$ has finite volume, it  is a straightforward computation to check that 
\begin{align}
\|\psi_{k}\|_{L^{2}}\to& \; 0\,,  \label{eq:psikL20}\\
\sup_{k} \|\psi_{k}\|_{W^{1,2}}<&\; \infty\,. \label{psikW12Bdd}
\end{align}
Since $W^{1,2}$ is a Hilbert space, the norm bound \eqref{psikW12Bdd} implies that we can extract a subsequence $\psi_{k'}$ which converges weakly in $W^{1,2}$ to some $\psi\in W^{1,2}$. Since $\psi_{k'}$ also converges weakly in $L^{2}$ to $\psi$, and we know from \eqref{eq:psikL20} that $\psi_{k'}$ converges  to $0$ strongly in $L^{2}$, we infer that $\psi=0$. So far we constructed a sequence $(\psi_{k'})\subset W^{1,2}$ converging to $0$ \emph{weakly} in $W^{1,2}$, where each $\psi_{k'}$ is equal to $1$ on a neighbourhood of the singular set $\{\rho=0\}$. 
By Mazur's Lemma, we can construct a sequence $(\phi_{j})\subset W^{1,2}$ of finite convex combinations of elements in  $(\psi_{k'})$ which converges to $0$ strongly in $W^{1,2}$. More precisely, there exists a function $N:\mathbb{N}\to \mathbb{N}$ and a sequence of sets of real numbers
$$
\{ \alpha(j)_{k'}\in [0,1] \colon k'=j, \ldots, N(j)\}
$$
with
$$
\sum_{k'=j}^{N(j)} \alpha(j)_{k'}=1\, ,
$$
such that the sequence $(\phi_{j})_{j\in \N}$ defined by the convex combination
$$
\phi_{j}=\sum_{k'=j}^{N(j)} \alpha(j)_{k'} \psi_{k'}
$$
converges strongly to $0$ in $W^{1,2}$. From its explicit expression,   it is clear that $\phi_j$ is equal to $1$ on a neighbourhood of the singular set $\{\rho=0\}$ (since each $\psi_{k'}$ has this property).
Thus, for $E=\{\rho=0\}$, the infimum in \eqref{eq:DefPolar} is zero and the set $\{\rho=0\}$ is polar.
\end{proof}

The  following consequence of the fact that $\mathcal{S}$ is polar is well known to experts. We give a self-contained proof for the reader's convenience.

\begin{proposition}\label{prop:W12=W120}
Let  $(X, \sfd, \mm)$ be an asymptotically D$p$-brane metric measure space satisfying the assumptions of Proposition \ref{Prop:SingPolar} and let  $\mathcal{S}$ be the singular set of $X$.
Let $W^{1,2}_{0,\mathcal{S}}(X,\di,\mm)$ be the closure (in $W^{1,2}$ topology) of the set of smooth functions compactly supported in $X\setminus \mathcal{S}$. Then $W^{1,2}_{0,\mathcal{S}}(X,\di,\mm)= W^{1,2}(X,\di,\mm)$; more precisely, the identity map $f\mapsto f$ is an isomorphism of Hilbert spaces between $W^{1,2}_{0,\mathcal{S}}(X,\di,\mm)$ and  $W^{1,2}(X,\di,\mm)$.
\end{proposition}

\begin{proof}
It is clear that the inclusion map $f\mapsto f$ from  $W^{1,2}_{0,\mathcal{S}}(X,\di,\mm)$ to  $W^{1,2}(X,\di,\mm)$ is an isometric immersion. Then it is enough to show that, for each $f\in W^{1,2}(X,\di,\mm)$ there exists a sequence $(\varphi_{k})$ of smooth functions with compact support in $X\setminus \mathcal{S}$ such that $\varphi_{k}\to f$ strongly in $W^{1,2}$. We prove the statement by subsequent approximations and conclude by a diagonal argument.

\textbf{Claim 1}. Let $f \in W^{1,2}(X,\di,\mm)$ and consider the sequence of truncations
\begin{equation*}
f_{n}(x)=
\begin{cases}
f(x)  &\text{ if } |f(x)|\leqslant n\\
n & \text{ if } f(x)\geqslant n \\
-n  &\text{ if } f(x)\leqslant -n\,\,.
\end{cases}
\end{equation*}
Then $f_{n}\to f$ strongly in $W^{1,2}$.
\\The claim follows directly by dominated convergence theorem.
\smallskip

Since by Proposition \ref{Prop:SingPolar} the singular set $\mathcal{S}$ is polar, then there exists a sequence of Lipschitz functions $(\psi_{k})$ with values in $[0,1]$, equal to $1$ on a neighbourhood of $\mathcal{S}$ and such that $\|\psi_{k}\|_{W^{1,2}}\to 0$. Up to extracting a subsequence not relabeled, we can also assume $\psi_{k} \to 0$ $\mm$-a.e.

\textbf{Claim 2}. Let $f \in W^{1,2}(X,\di,\mm)\cap L^{\infty}(X,\mm)$. Then  $(1-\psi_{k}) f$ is a $W^{1,2}(X,\di,\mm)$ function supported in the regular part  $X\setminus \mathcal{S}$, for each $k\in \N$; moreover, the sequence  $((1-\psi_{k}) f)_{k}$ converges to $f$ strongly in $W^{1,2}$.
\\This last claim is equivalent to show  that $\|\psi_{k}f\|_{W^{1,2}}\to 0$ as $k\to \infty$. Since $|\psi_{k} f|\leqslant f$ and $\psi_{k} f\to 0$ a.e., by dominated convergence theorem it follows that $\|\psi_{k}f\|_{L^{2}}\to 0$. It is thus sufficient to show that $\|\nabla(\psi_{k}f)\|_{L^{2}}\to 0$. We have
$$
\int_{X} | \nabla(\psi_{k}f)|^{2} \dd\mm\leqslant 2  \int_{X} | \nabla \psi_{k}|^{2} f^{2} \dd\mm + 2  \int_{X} | \psi_{k}|^{2} |\nabla f|^{2} \dd\mm\,.
$$
 Since $f$ is bounded and $\|\psi_{k}\|_{W^{1,2}}\to 0$, then the first  integral in the right hand side converges to zero as $k\to \infty$.
The second integral in the right hand side converges to zero as $k\to \infty$ by dominated convergence theorem, since $| \psi_{k}|^{2} |\nabla f|^{2}\leqslant |\nabla f|^{2}\in L^1(X,\mm)$. 
\smallskip

\textbf{Claim 3}. If $\phi \in W^{1,2}(X,\di,\mm)$ has support contained in the regular part  $X\setminus \mathcal{S}$, then there exists a sequence $(\phi_{j})$ of smooth functions with compact support in $X\setminus \mathcal{S}$ such that $\phi_{j}\to \phi$ strongly in $W^{1,2}$.
\\This last claim is completely standard and can achieved by using partition of unity to localise in coordinate charts and then use approximation by convolution in each chart to obtain smooth functions; finally multiplying by smooth cut-off functions with compact support in $X\setminus \mathcal{S}$ gives the desired approximation $(\phi_{j})$.
\smallskip

We can now combine the three claims above to conclude. Let $f \in W^{1,2}(X,\di,\mm)$ and fix $\epsilon>0$. We will construct $\tilde{f}$ smooth with compact support in $X\setminus \mathcal{S}$ such that
\begin{equation}\label{eq:f-tildef}
 \|f- \tilde{f}\|_{W^{1,2}}\leqslant \epsilon\, .
 \end{equation}
By claim 1, there exists $f_{n}\in W^{1,2}(X,\di,\mm)\cap L^{\infty}(X,\mm)$ such that 
\begin{equation}\label{eq:f-fn}
  \|f-f_{n}\|_{W^{1,2}}\leqslant \epsilon/3\, .
\end{equation}
 By claim 2, there exists $\phi_{k,n}= (1-\psi_{k}) f_{n}\in W^{1,2}(X,\di,\mm)$ with compact support in  the regular part  $X\setminus \mathcal{S}$  such that
 \begin{equation}\label{eq:fn-phikn}
 \|f_{n}- \phi_{k,n}\|_{W^{1,2}}\leqslant \epsilon/3\, .
 \end{equation}
 Finally, by claim 3, there exists $\tilde{f}$ smooth with compact support in $X\setminus \mathcal{S}$ such that 
  \begin{equation}\label{eq:phikn-tif}
\|\phi_{k,n}-\tilde{f}\|_{W^{1,2}}\leqslant \epsilon/3\, . 
 \end{equation}
 The combination of \eqref{eq:f-fn},  \eqref{eq:fn-phikn}, \eqref{eq:phikn-tif} gives \eqref{eq:f-tildef} by triangle inequality.
\end{proof}

\begin{remark}\label{rem:W12-polar}
 It is a general fact (well known to experts) for metric measure spaces that if $\mathcal{S}\subset X$ is polar than $W^{1,2}(X,\sfd,\mm)$ coincides with the closure in $W^{1,2}$-topology of the set of $W^{1,2}$-functions with support contained in $X\setminus \mathcal{S}$. The proof in the general case can be obtained along the lines of the proof of Proposition \ref{prop:W12=W120}.
\end{remark}

As observed above, if  $(X,\di,\mm)$ is an asymptotically D$p$-brane metric measure space, then the Sobolev space $W^{1,2}(X,\sfd,\mm)$ is a Hilbert space and we are in the framework described in Sec.\;\ref{sec:spectral}.

Given a closed subset $\mathcal{S}\subset X$, the \emph{Laplacian with Dirichlet boundary conditions on $\mathcal{S}$} is the analog of the construction performed in Sec.\;\ref{sec: PEMS} for the definition of $\Delta$, replacing $W^{1,2}(X,\sfd,\mm)$ by  $W^{1,2}_{0,\mathcal{S}}(X,\di,\mm)$. We denote this Dirichlet Laplacian by $\Delta_0$ and view it as an operator in its associated domain $D(\Delta_0)$, in the sense specified in Sec.\;\ref{sec:spectral}.

\begin{remark}[The spectrum of the Dirichlet Laplacian is always contained in the spectrum of the Laplacian]
Let  $(X,\sfd,\mm)$ be a metric measure space such that $W^{1,2}(X,\sfd,\mm)$ is a Hilbert space and let $\mathcal{S}\subset X$ be a closed subset. In this general situation, where smooth functions are not necessarily at disposal, one can define $W^{1,2}_{0,\mathcal{S}}(X,\di,\mm)$ to be the closure in $W^{1,2}$-topology of the set of $W^{1,2}$-functions in $X$ with essential support\footnote{Recall that for a measurable function $f$ defined on $(X,\mm)$ the essential support is the smallest closed subset $E\subset X$ such that $f(x)=0$ $\mm$-a.e.~outside $E$.} contained in $X\setminus \mathcal{S}$.  Denote by $\Delta$ the Laplacian of $(X, \sfd, \mm)$ and let $\Delta_{0}$ be the Laplacian on $X\setminus \mathcal{S}$ with Dirichlet boundary conditions on $\mathcal{S}$. 
 By the simple fact that  $W^{1,2}_{0,\mathcal{S}}(X,\di,\mm)$ can be seen as a closed sub-space of $W^{1,2}(X,\sfd,\mm)$, it follows that  $\sigma(\Delta_0)\subset \sigma(\Delta)$; moreover, if $f$ is an eigenfunction with eigenvalue $\lambda$ of $\Delta_0$ then  $f$ is an eigenfunction with eigenvalue $\lambda$ of $\Delta$.
 \end{remark}

In the case of D$p$-brane m.m.s., since by Prop.~\ref{prop:W12=W120} we know that $W^{1,2}_{0,\mathcal{S}}(X,\di,\mm)$ coincides with $W^{1,2}(X,\sfd,\mm)$, the  Laplacian with Dirichlet boundary conditions on $\mathcal{S}$ coincides with the Laplacian of $(X,\sfd,\mm)$ as defined in Def.~\ref{def:Laplacian}. We therefore obtain the following corollary.

\begin{corollary}\label{Cor:Delta=Delta0}
Let  $(X, \sfd, \mm)$ be an asymptotically D$p$-brane metric measure space satisfying the assumptions of Proposition \ref{Prop:SingPolar}. Let $\mathcal{S}$ be the singular set of $X$. Let $\Delta$ be the Laplacian of $(X, \sfd, \mm)$ and let $\Delta_{0}$ be the Laplacian on $X\setminus \mathcal{S}$ with Dirichlet boundary conditions on $\mathcal{S}$. 

Then $\Delta$ and $\Delta_{0}$ have exactly the same spectral properties, i.e.~$\rho(\Delta)=\rho(\Delta_{0}), \sigma_d(\Delta)=\sigma_d(\Delta_{0}),\sigma_{ess}(\Delta)=\sigma_{ess}(\Delta_{0})$.
 
In particular, in the variational characterization  \eqref{eq:defeigk} of $\lambda_{k}$, one can assume  $V_{k}$  to be a $k$-dimensional subspace of the vector space of smooth functions with compact support contained in the regular part $X\setminus \mathcal{S}$.
\end{corollary}

\begin{remark}
	For some singularities we expect we should include also functions with Neumann boundary conditions; for example one can argue this for a single D6, using duality with M-theory \cite[Sec.~3.3]{passias-t}. However, when the singularity $\mathcal{S}$  is polar, the eigenvalue problems for $\Delta$ with Dirichlet boundary conditions and with Neumann boundary conditions on $\mathcal{S}$ are equivalent (in turn to the eigenvalue problem without boundary conditions. This is due to the fact that a polar set is ``invisible'' by $W^{1,2}$-functions). Indeed, it is possible to approximate arbitrarily well functions with Neumann boundary conditions with functions in $W^{1,2}_{0,\mathcal{S}}(X,\di,\mm)$, given Rem.~\ref{rem:W12-polar} and Prop.~\ref{prop:W12=W120}.
\end{remark}

We have found that we can retrieve the metric-measure Laplacian also by working on the space $X\setminus {\mathcal S}$ without singular locus, at least when the latter is polar. In particular, we already have a domain on which the Laplacian is self-adjoint: namely the finiteness domain 
\begin{equation}\label{eq:defDDelta}
D(\Delta):=\{f\in L^2(X,\mm) \colon \Delta f \in L^2(X,\mm)\}
\end{equation}
that we have introduced in Def.\;\ref{def:Laplacian}, endowed with the norm 
$$
\|f\|_{D(\Delta)}:= \left(\int_X |f|^2+ |\Delta f|^2 \dd \mm\right)^{1/2}.$$ In a sense there is no need to extend the domain, using the terminology of Sec.~\ref{sec:spectral}.

We can nevertheless explore alternatives, and a first possibility is inspired by the discussion of the self-adjoint extension of the Hamiltonian $-\partial^2_z-1/4z^2$ in Sec.~\ref{sub:pot}.\footnote{A recent study of the influence of the choice of domain on KK stability is in \cite{mourad-sagnotti-selfadjoint}. It would be interesting to revisit that model with our methods.} However, some complications appear from this perspective. For illustrative purposes, let us assume here that the singular set ${\mathcal S}$ reduces to one point $x_{0}$. Let $G\in L^{1}_{loc}(X)$ be the Green's function for the weighted Laplacian on the original space $X$, centred at $x_{0}$: $\Delta G = \delta_{x_{0}}$.  Working on $X\setminus\{x_{0}\}$, one might think $G$ becomes harmonic; thus if one manages to include $G$ in the domain of the Laplacian while keeping it self-adjoint, then $G$ would be an eigenfunction with zero eigenvalue. This idea however presents some challenges.
	
	First of all, while $G$ satisfies $\Delta G=0$ in a  \emph{point-wise} sense on $X\setminus\{x_{0}\}$, in the  \emph{distributional} sense   $G$ in fact still satisfies $\Delta G= \delta_{x_{0}}$. Of course this cannot be seen by testing against functions in $C^{\infty}_{c}(X\setminus \{x_{0}\})$, as they vanish at $x_{0}$. However  for every test function $f$ continuous on $X\setminus \{x_{0}\}$ with bounded support, with $ \Delta f\in L^\infty(X)$, and admitting a limit $\lim_{x\to x_{0}} f(x)$, it would hold
	$$ \lim_{x\to x_{0}} f(x)=\int_{X} G \Delta f \, \dd \mm = \int_{X\setminus\{x_{0}\}} G \Delta f \, \dd \mm =\int_{X\setminus\{x_{0}\}} f \Delta G \, \dd \mm;$$
	 the first equality comes by direct computation as by construction $G$ is the Green function,  the second uses the fact that $\{x_{0}\}$ has measure zero, and the last uses the assumption that $\Delta$ is self-adjoint with $G$ in its domain. 
	 
	 Even if  somehow one managed to impose that $G$ is in the domain of $\Delta$ with $\Delta G =0$ on $X\setminus\{x_{0}\}$, a second problem would appear. Now $\int_{X\setminus\{x_{0}\}} G \Delta G \, \dd \mm =0$, but on the other hand $\int_{X\setminus\{x_{0}\}} |\nabla G|^2 \dd \mm$ is clearly non-zero (and in fact diverges when $\{x_{0}\}$ has codimension $\geqslant 2$). This means that integration by parts would no longer be valid on the domain of $\Delta$, which in turn invalidates the variational approach to the Laplace spectrum, in terms of Rayleigh quotients \eqref{eq: Rayleigh}.
	 
	 A third challenge is that  $X\setminus\{x_{0}\}$ is not geodesically complete  (unless $x_0$ is at infinite distance). It is a classical result \cite{Gaffney, Roelcke} that geodesic completeness of a smooth Riemann manifold $(M,g)$ is equivalent to the essential self-adjontness of the Laplace Beltrami operator on $C^{\infty}_c(M)$; in turn, essential self-adjontness is a key assumption in spectral theory.

\begin{remark}
	 In case ${\mathcal S}$ has null $W^{2,2}$-capacity, we can also argue that the extension domain for $\Delta$ we have chosen in \eqref{eq:defDDelta}  is unique, in the following sense. Let $C^\infty_\mathrm{c}(X \setminus {\mathcal S})$ be the space of functions with compact support outside the singular set ${\mathcal S}$. 
		Suppose we want to obtain a Hilbert space $H$ in which  $C^\infty_\mathrm{c}(X \setminus {\mathcal S})$  is dense and such  that $\Delta: H\to L^2(X)$ is a densely defined, self-adjoint operator. Being self-adjoint, $\Delta$ is automatically closed and this condition forces to endow   $C^\infty_\mathrm{c}(X \setminus {\mathcal S})$ with the norm induced by the quadratic form ${\mathcal Q}(f):=\int (f^2 + |\Delta f|^2) \dd \mm$ (or an equivalent norm inducing the same topology). $H$ then coincides with the closure of   $C^\infty_\mathrm{c}(X \setminus {\mathcal S})$  in $D(\Delta)$,  which in turn coincides with $D(\Delta)$ as in \eqref{eq:defDDelta}  since by assumption ${\mathcal S}$ has null $W^{2,2}$-capacity. (For more details on $W^{2,2}$-capacity and on sufficient conditions for a set to have null $W^{2,2}$ capacity in a metric measure space, we refer to \cite{Hinz2023}.) Notice this is surely the case for  an asymptotically D$p$-brane metric measure space in the sense of Def.\;\ref{def: Dp-brane mms}, for $p=1,\ldots, 5$ (as the singular set is at $\infty$); for $p\geq 6$ we expect this is not true as the singular set has not enough co-dimension. (One would need co-dimension 4 in a suitable weighted Hausdorff sense; we do not delve in more details as we do not expect a positive result.) \end{remark}

\section{Gravity localization in string theory} 
\label{sec:st}
So far we have discussed general mechanisms for gravity localization in theories with extra dimensions. As we have seen, even when the internal warped volume is infinite, there can still be meaningful localization of gravity on a lower-dimensional subspace, such as in the KR mechanism reviewed in Sec.~\ref{sub:ads}.

In this section we consider realizations in string/M-theory of this mechanism, focusing on the AdS case. In Sec.~\ref{sub:genmetloc}, we review some bounds on the lowest KK masses $m_0\neq 0$ and $m_1$, coming from \cite{deluca-deponti-mondino-t}. We will prove a general result on absence of separation of scale for theories with massive gravitons in AdS when only energy sources that satisfy the Reduced Energy Condition are turned on in the background (and when there is an upper bound on the gradients of the warping). Our general theorems also allow to infer localization of gravity without computing the spectrum.

Localization can happen on a brane such as in \cite{verlinde-RS2,chan-paul-verlinde} or on a broader internal region, loosely referred to as ``thick brane''. In particular, in Sec.~\ref{sub:Riemann} we construct realizations of massive gravity in String/M-theory starting from solutions that contain Riemann surfaces, with or without supersymmetry. In Sec.~\ref{sub:BL} we study instead models with $\mathcal{N}$ = 4 supersymmetry in type IIB string theory.

When separation of scales is absent, knowledge of the eigenvalues allows to put a lower bound on localization of gravity of the scale of the four-dimensional cosmological constant.  Whether localization also holds at smaller scales depends on the behavior of the eigenfunctions, and we will show how for the models of Sec.~\ref{sub:BL} the situation might indeed be better.

\subsection{General method and a bound on scale separation}\label{sub:genmetloc}

Our goal is to find vacuum solutions (namely, spacetimes of the form (\ref{eq:vac})) with infinite warped volume, so that the massless graviton is not part of the spectrum, and where the first massive spin 2 field with mass $m_0$ is very light and separated from the rest of the tower starting with mass $m_1$. This will guarantee localization of gravity at least up to the scale $m_1$. 

As we illustrated with the five-dimensional models studied in Sec.~\ref{sub:ads}, we can obtain this hierarchy by appropriately tuning the Cheeger constants. Recall that the first two generalized Cheeger constants are given respectively by
\begin{subequations}\label{eq:h-inf}
\begin{equation}\label{eq:h0-inf}
	h_0 := \inf_{B_0} \frac{\text{Vol}_A(\partial B_0)}{\text{Vol}_A(B_0)}\, , 
\end{equation}
and
\begin{equation}\label{eq:h1-inf}
	h_1 := \inf_{B_0,B_1} \max\left\{\frac{\text{Vol}_A(\partial B_0)}{\text{Vol}_A(B_0)},\frac{\text{Vol}_A(\partial B_1)}{\text{Vol}_A(B_1)} \right\}\, ,
\end{equation}
\end{subequations}
where the infimum is taken with respect to all possible (disjoint) sets $B_0,B_1$ of finite volume and the subscript $A$ refers to the fact that volumes are weighted with $e^f = e^{(D-2)A}$. Using the results in \cite[Th.~4.2, 4.8]{deluca-deponti-mondino-t} and \cite[Th.~6.7, 6.8]{deluca-deponti-mondino-t-entropy}, we find that these constants control the first two eigenvalues of $\Delta_f$ as
\begin{equation}\label{eq:bound0InfVol}
	\frac14 h_0^2 \leqslant m_0^2 < \max\left\{\frac{21}{20} \sqrt{-K}h_0,\frac{22}{20}h_0^2 \right\}\, ,
\end{equation}
and
\begin{equation}\label{eq:bound1InfVol}
	C\cdot h_1^2 \leqslant m_1^2< \max\left\{-\frac{3528}{25}K,\frac{3696}{25}\sqrt{-K}h_1, \frac{3872}{25}h^2_1\right\}\, ,
\end{equation}
where $K\le 0$ is a lower bound on the $\left(\text{Ricci}_f^\infty\right)_{mn} := R_{mn}-\nabla_m\nabla_nf$ and $C>0$ is a universal constant. The bounds \eqref{eq:bound0InfVol}, \eqref{eq:bound1InfVol} are valid also in presence of some singularities: the lower bounds in the infinitesimally Hilbertian class of Sec.~\ref{sec: PEMS}, which includes all the physical spaces where the closure of the singular set has measure zero; the upper bounds in the smaller RCD class (recall footnote \ref{foot:CD} and discussion below), that includes D-branes \cite[Sec.~3]{deluca-deponti-mondino-t}.
Using \cite[Th.~4.4]{deluca-deponti-mondino-t} we can also directly relate the first two eigenvalues to each other, without any reference to the Cheeger constants:
\begin{equation}\label{eq: eigen0 mainbound K<0inf}
	m_1^2<\max\left\{-\frac{3528}{25}K,\frac{704}{5}m_0^2\right\}.
\end{equation}

We are interested in a limit in which $m_0 \ll 1$ with $m_1$ remaining finite, so that $m_1/m_0 \gg 1$ and the rest of the tower is separated from the almost massless graviton responsible for localizing gravity. From \eqref{eq:bound0InfVol}, \eqref{eq:bound1InfVol} we see that this is achieved in a limit in which $h_0\to 0$ while $h_1$ stays finite.  The estimate \eqref{eq: eigen0 mainbound K<0inf} implies that in this limit 
\begin{equation}\label{eq:inequalitym1K}
	m_1^2<-\frac{3528}{25}K\,.
\end{equation}
A general result on $K$ was proven in \cite{deluca-t-leaps}. On any reduction of any higher-dimensional gravitational theory that at low-energies reduces to Einstein gravity plus some matter content (encoded in a stress-energy tensor $T_{M N}$), the equations of motion impose
\begin{equation}
	(\text{Ricci}_f^\infty)_{mn}= \Lambda_4 g_{m n}-(D-2)\partial_m A\partial_n A+\frac{1}{2}\kappa^2\left(T_{m n}-\frac{1}{d} T^{(d)}\right)\,,
\end{equation}
where $m,n$ are internal directions, and $T^{(d)}$ denotes the trace of $T_{M N}$ along the $d$-dimensional vacuum. When $\left(T_{m n}-\frac{1}{d} T^{(d)}\right)\geqslant 0$, a condition named \emph{Reduced Energy Condition} (REC) in \cite{deluca-t-leaps}, it is then possible to express the lower bound $K$ on $\text{Ricci}_f^\infty$ uniquely in terms of $\Lambda_4$ and $|\dd A|$. In particular, the REC has been shown to hold for a variety of matter content, including scalar fields, $p$-form fluxes, higher dimensional cosmological constants, localized sources with positive tension and general scalar potentials. Further specializing to $\Lambda_4 < 0$, we thus have
\begin{equation}\label{eq:KAdS}
	K = -\left(|\Lambda_4| + (D-2) (\textrm{sup}|\dd A|)^2\right).
\end{equation}
We stress that this is true for any higher dimensional gravitational theory when only sources that satisfy the REC are turned on. In particular, this holds true for string theory solutions without O-planes nor quantum effects. In such situations we have the following
\begin{proposition}\label{prop:Sep}	
	In any AdS vacuum solution with infinite warped volume that satisfies the Reduced Energy Condition and that admits a very light spin 2 field of mass $m_0 \ll |\Lambda_4|$ in its spin 2 Kaluza--Klein tower, the second element of the tower is bounded by 
	\begin{equation}\label{eq:inequalitym1Lambda}
		\frac{m_1^2}{|\Lambda_4|}<\frac{3528}{25}\left(1+ (D-2) \frac{(\mathrm{sup}|\dd A|)^2}{|\Lambda_4|}\right)\,.
	\end{equation}
\end{proposition}

Summarizing, when $|\dd A|$ is bounded, Prop.~\ref{prop:Sep} proves a bound on separation of scales, defined in this case as the hierarchy between $m_1$ and $|\Lambda_4|$. In particular, separation of scales is forbidden in compact spaces when $\textrm{sup}|\dd A|$ is not allowed to grow much higher than $|\Lambda_4|$. This condition is natural as seen from the fact that the solution has to be in the supergravity approximation. Indeed, generically, a very large $|\dd A|$ is expected to result in a very large (in absolute value) $D$-dimensional curvature.

\begin{remark}
	We have specialized Prop.~\ref{prop:Sep} to AdS, but a similar statement also holds for dS vacuum solutions. Indeed, while for compact internal spaces it is known that the REC implies a negative cosmological constant \cite{gibbons-nogo,dewit-smit-haridass,maldacena-nunez}, for more general warped products this is not true.\footnote{Simple examples with infinite warped volume can be obtained by regarding any AdS$_n$ compactification as a non-compact foliation of $dS_{n-1}$. For more elaborated constructions in 10/11 dimensional supergravities see \cite{gibbons-hull-ds}.} Allowing for $\Lambda_4 > 0$,  \eqref{eq:KAdS} changes by $|\Lambda_4| \to - \Lambda_4$ and simple modification of Prop.~\ref{prop:Sep} is obtained by removing the ``1'' in \eqref{eq:inequalitym1Lambda}.
	However, since now $K$ can change sign, it is possible to obtain stronger results for this case by using the theorems formulated for $K\geqslant 0$. More generally, allowing also for REC-violating sources, $K$ will have extra negative contributions from those. Terms with different signs can then compensate each other and we leave the exploration of these richer scenarios to future work. 
\end{remark}

From the point of view of gravity localization, we have obtained that looking at the eigenvalues alone can guarantee localization at most up to the scale of the synthetic Ricci lower bound $K= -\left(|\Lambda_4| + (D-2) (\textrm{sup}|\dd A|)^2\right)$. Whether actual localization can be pushed to higher energy scales depends on the relative concentration of eigenfunctions, as in Sec.~\ref{sub:ads}, and it is not readily visible from the eigenvalues alone.

In the next two subsections, we will evaluate these bounds on two different classes of String/M-theory AdS vacua showing how gravity is localized in such UV complete examples.

We conclude with a brief discussion of the continuous spectrum. Since we are considering noncompact ``internal'' spaces $X_n$, the spectrum is not guaranteed to be discrete. Some general methods to analyze this issue were presented in \cite{cianchi-mazya,wang-infinite}. For our purposes, we should make sure that the continuous spectrum is either absent, or starts at a large value.

Already when the weighted volume (and hence $m_4$) is finite, D$p$-branes with $p\le 5$ are at infinite distance; so in their presence discreteness is not guaranteed. We presented a first rough analysis in \cite[Sec.~4.2.2]{deluca-deponti-mondino-t}: near the sources, the ratios in (\ref{eq:h-inf}) get arbitrarily small for $p<5$, tend to $r_0^{-1}$ for $p=5$, and get large for $p>5$. This suggests that for $p<5$ the continuum is present and starts at zero, for $p=5$ it starts at $r_0^{-1}\sim l_s^{-1} g_s^{-1/2}$, and for $p>5$ it is absent. A more sophisticated analysis can be carried out using \cite[Th.~3.1]{cianchi-mazya}. This gives a characterization of spaces where the ordinary Laplace--Beltrami operator has a discrete spectrum, in terms of a certain ``isocapacitary'' function $\mu_{X_n}(s): [0,\infty)\to [0,\infty)$. We expect this theorem to still hold in the weighted case. We checked that its hypotheses fail near a D$p$-brane, $p\leqslant 5$, and hold for $p>5$.

In the models we are going to see, the weighted volume is infinite, and the analysis above is not enough: in other words, branes are not the only source of non-compactness. We will return to the issue in each case separately.

\subsection{Riemann surfaces} 
\label{sub:Riemann}
Examples of infinite-volume vacuum solutions that localize gravity in lower dimensions can be readily constructed starting from any compactification that includes a Riemann surface $\Sigma_g$, when the background fields are constant along $\Sigma_g$. Known four-dimensional examples include the AdS$_4$ solutions in type IIB string theory, that can be found by combining \cite{cvetic-lu-pope-F4,nunez-park-schvellinger-tran}, as spelled out in  \cite[Sec.~5.2]{legramandi-nunez}; and examples in massive type IIA \cite{nunez-park-schvellinger-tran,bah-passias-weck}.\footnote{These massive type IIA examples include the presence of O8-planes, which are outside of the $\mathsf{RCD}$ class but included in the class of infinitesimally Hilbertian spaces introduced in Sec.~\ref{sec: PEMS}. Thus, for these examples, only a subset of the eigenvalue bounds is currently known to hold. We refer the reader to \cite[Sec.~6.3]{deluca-deponti-mondino-t-entropy} for more details.} Examples in other dimensions include the AdS$_5$ constructions \cite{maldacena-nunez, bah-beem-bobev-wecht} in M-theory and \cite{afpt} in type IIA, or the AdS$_3$ vacua in type IIB of \cite{benini-bobev,couzens-macpherson-passias}.

In general, we are considering solutions where the internal unwarped metric is a fibration over a Riemann surface:
\begin{equation}
	\dd s^2_n = L_n^2\left(\dd s^2_{\Sigma_g} + g_{i j} (\dd y^i+A^i)(\dd y^j+A^j)\right)\,,
\end{equation}
where $\dd s^2_{\Sigma_g}$ is the metric on a negatively-curved Riemann surface of genus $g$ with Ricci scalar $-2$, and $A^i$, $i = 1,\dots n-2$ are a collection of one-forms on $\Sigma_g$. The weight function $e^f$ and the metric $g_{ij}$ are constant along $\Sigma_g$; $L_n$ is a constant that in these models is often of order $\sim 1/|\Lambda_4|$.

While the spin 2 fields are given by eigenfunctions of the weighted Laplacian $\Delta_f$, a subset of modes that are non-constant only along $\Sigma_g$ has masses dictated by the spectrum of the standard Laplacian on the Riemann surface. We can thus ask whether this portion of the spectrum can include a very light mode, specifically when the volume of $\Sigma_g$ is infinite so that no massless graviton is allowed, as we proved in Sec.~\ref{sec:const}.
\begin{wrapfigure}[14]{L}{7cm}
	\centering
		\includegraphics[width=7cm]{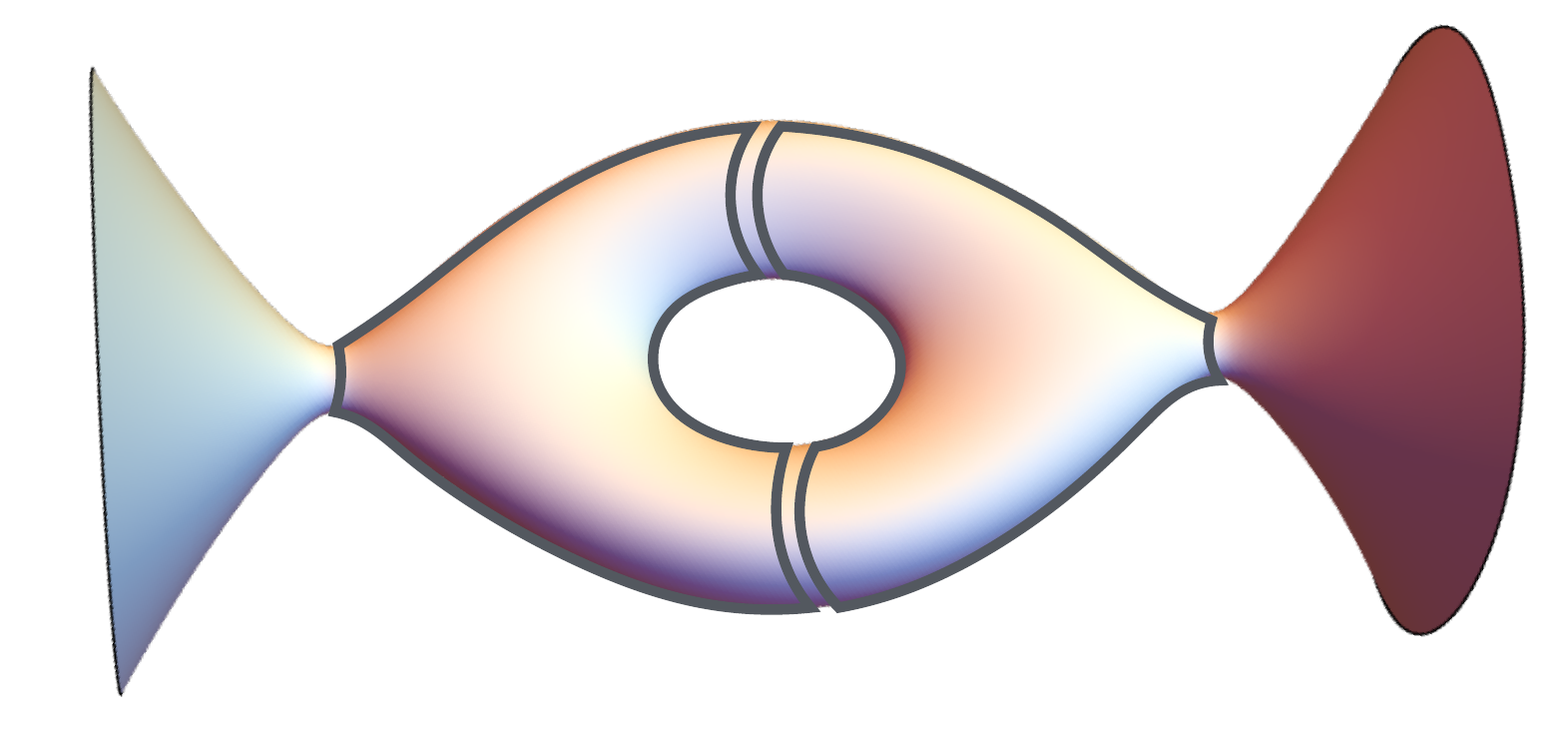}
	\caption{\small An infinite-area Riemann surface ($g = 1, n_f = 2$), with a single small eigenvalue. $h_0$ is small as can be seen by taking as $B_0$ the union of the two pair of pants; $h_1$ is not small since two of the three ends of the $B_i$ are not small.}
	\label{fig:torus-funnel}
\end{wrapfigure}

\noindent The answer is affirmative, and there is considerable freedom in tuning the geometry to achieve a hierarchy between the first mode and the rest of the tower. This phenomenon can be directly understood from the general bounds in terms of the Cheeger constants by using the equations \eqref{eq:bound0InfVol} and \eqref{eq:bound1InfVol}. To see this, we specialize those general inequalities to the spectrum of the pure Laplacian on $\Sigma_g$ by setting $A=f = 0$ and $K = -1$. 
In this case, there is a natural choice for how to construct the $B_i$, by using the pieces provided by the \emph{pair of pants} decomposition of $\Sigma_g$. Specifically, an infinite-area negatively-curved  $\Sigma_g$, with finite Euler characteristic and no cusps, can always be decomposed in a \emph{compact core} plus $n_f$ \emph{funnels} and then cut in $2g-2+n_f$ pair of paints.\footnote{We refer the reader to \cite[Sec.~6.3]{deluca-deponti-mondino-t} and references therein for more details on infinite area Riemann surfaces and their spectra.} Importantly, each of these pieces can have ends of arbitrary geodesic length. Thanks to this freedom, we can obtain arbitrarily small  $\frac{\text{Vol}(\partial B)}{\text{Vol}(B)}$ for each of the pair of pants in the decomposition, thus realizing arbitrarily small $h_i$ for $i<2g-2 +n_f$. In particular, we are free to tune $h_0$ to be very small while keeping $h_1$ fixed, thus realizing the mechanism described in Sec.~\ref{sub:genmetloc}.  An example with $g = 1, n_f = 2$, with a single light mode, is given in Fig.~\ref{fig:torus-funnel}. 

So far, we have discussed the part of the spectrum coming from modes only varying along $\Sigma_g$. While this is enough to ensure that a very light $m_0$ is part of the spectrum, we would like to ensure that $m_1$ does not get small, regardless of its origin. If the geometry doesn't have other small ``necks'', so that $h_1$ computed in the whole geometry is not small, then this is guaranteed by the first inequality in \eqref{eq:bound1InfVol}. It is easy to check this explicitly case by case for each example: the fiber metric $g_{ij}$ is that of a distorted sphere of cohomogeneity one, independent of $\Sigma_g$. For the original AdS$_5$ example of \cite{maldacena-nunez}, this check was performed in \cite[Sec.~6]{deluca-deponti-mondino-t}.

As we noted in our general discussion in section \ref{sub:genmetloc}, when $X_n$ is non-compact the spectrum is not guaranteed to be discrete. For the models of this section, it is indeed known that infinite-area negatively-curved hyperbolic surfaces have a continuous spectrum that in our normalization starts at $1/4L_n^2$.

While this shows that gravity in these models is localized up to a length scale $\sim 1/m_1$, as in the KR model the true scale at which gravity looks four-dimensional might be even smaller if the eigenfunctions are sufficiently localized. The reason this might work is the following. We expect the matter modes $\lambda$ of which a four-dimensional observer is made to be concentrated on the compact core of $\Sigma$.
The expression (\ref{eq:massive-V}) would now be replaced by $(G M_1 M_2/R) \sum_{k=0}^\infty \ee^{-m_k R} \langle\lambda, \psi_k \rangle^2$, where $\langle \lambda, \psi\rangle:= \int_{M_6} \sqrt{g_6}\ee^{8A} \lambda \psi$. Unfortunately finding $\lambda$ would require obtaining the rest of the KK tower; the relevant differential operators ${\mathcal O}$ on $M_6$ have not been worked out for this class of solutions. Qualitatively, one would expect  $\lambda$ to be concentrated near the bulb just as $\psi_0$ is, and hence $\langle\lambda, \psi_{k>0}\rangle := \epsilon$ to be small, since $\langle \psi_{k>0}, \psi_0 \rangle=0$. For two such particles, 
\begin{equation}\label{eq:V-hope}
	V\sim G M_1 M_2 \left(\frac1R + \frac{\epsilon L_4}{R^2} \right)\,;
\end{equation}
thus for such modes we would find localization at a lower scale, $R> \epsilon L_4$. It would be interesting to check explicitly this conjectural mechanism.

An alternative to this mechanism would be to include localized sources on $\Sigma_g$, such as the so-called punctures. Now four-dimensional matter might be localized on such punctures, and the eigenfunctions corresponding to $m_1$ and the higher modes might be suppressed away from the sources where the matter is localized. It would be interesting to explore this further.  

Finally, we notice that one might try to avoid the need to study the eigenfunctions to have localization up to shorter distances by trying to push directly $m_1$ to much larger scales. However Prop.~\ref{prop:Sep} forbids this if $\sup|\dd A |$ is of order 1. 
This is the case for the examples of solutions we mentioned above, and we expect it to be true in general for this class of solutions since the warping does not depend on $\Sigma_g$.

\subsection{Examples with ${\mathcal N}=4$ supersymmetry} 
\label{sub:BL}

There exists a class of models that shares many similarities with the AdS version of the KR model of Section \ref{sub:ads}. It is holographically dual to domain walls in ${\mathcal N}=4$ super-Yang--Mills \cite{dhoker-estes-gutperle}. Its potential use for gravity localization was explored in \cite{bachas-estes,assel-phd,bachas-lavdas2}. One expects many similar realizations in string theory of such domain walls, and more generally of conformal defects within a CFT.

We refer the reader to the papers above for a detailed introduction to these solutions; our notation is based on our earlier \cite[Sec.~5]{deluca-deponti-mondino-t}. The geometry is again a warped product
\begin{equation}\label{eq:wp}
	\dd s^2 = \ee^{2A}(\dd s^2_{\mathrm{AdS}_4} + \dd s^2_{M_6})\,.
\end{equation}
The internal space is a fibration  over $\mathbb{R}_x$ with coordinate $x$, whose generic fibers are topologically $S^5$. The latter has the metric of a ``join'', with an $S^2\times S^2$ fibred over an interval $y\in I_y=[0,\pi/2]$, each $S^2$ shrinking at one of its endpoints. The solution is completely determined by two harmonic functions $H_1$, $H_2$ on the strip $\mathbb{R}_x \times I_y$.
Morally the coordinate $x$ will play the role of $r$ in the five-dimensional KR model. It is possible to include arbitrary numbers of NS5- and D5-branes along an AdS$_4\times S^2$, respectively at $y=0$ and $\pi/2$. The overall topology of $M_6$ was initially taken to be $\mathbb{R} \times S^5$, but appropriate choices for $H_1$, $H_2$ were later found to make the $S^5$ shrink at $x\to -\infty$, leading to topology $M_6\cong\mathbb{R}^6$ \cite{assel-bachas-estes-gomis,aharony-berdichevsky-berkooz}; or at both $x\to \pm \infty$, leading to $M_6 \cong S^6$ \cite{assel-bachas-estes-gomis}. A further possibility is to identify $\mathbb{R}$ periodically, leading to $M_6\cong S^1 \times S^5$ \cite{assel-bachas-estes-gomis2}.

\begin{wrapfigure}[13]{R}{6cm}
	\centering
		\includegraphics[width=5.5cm]{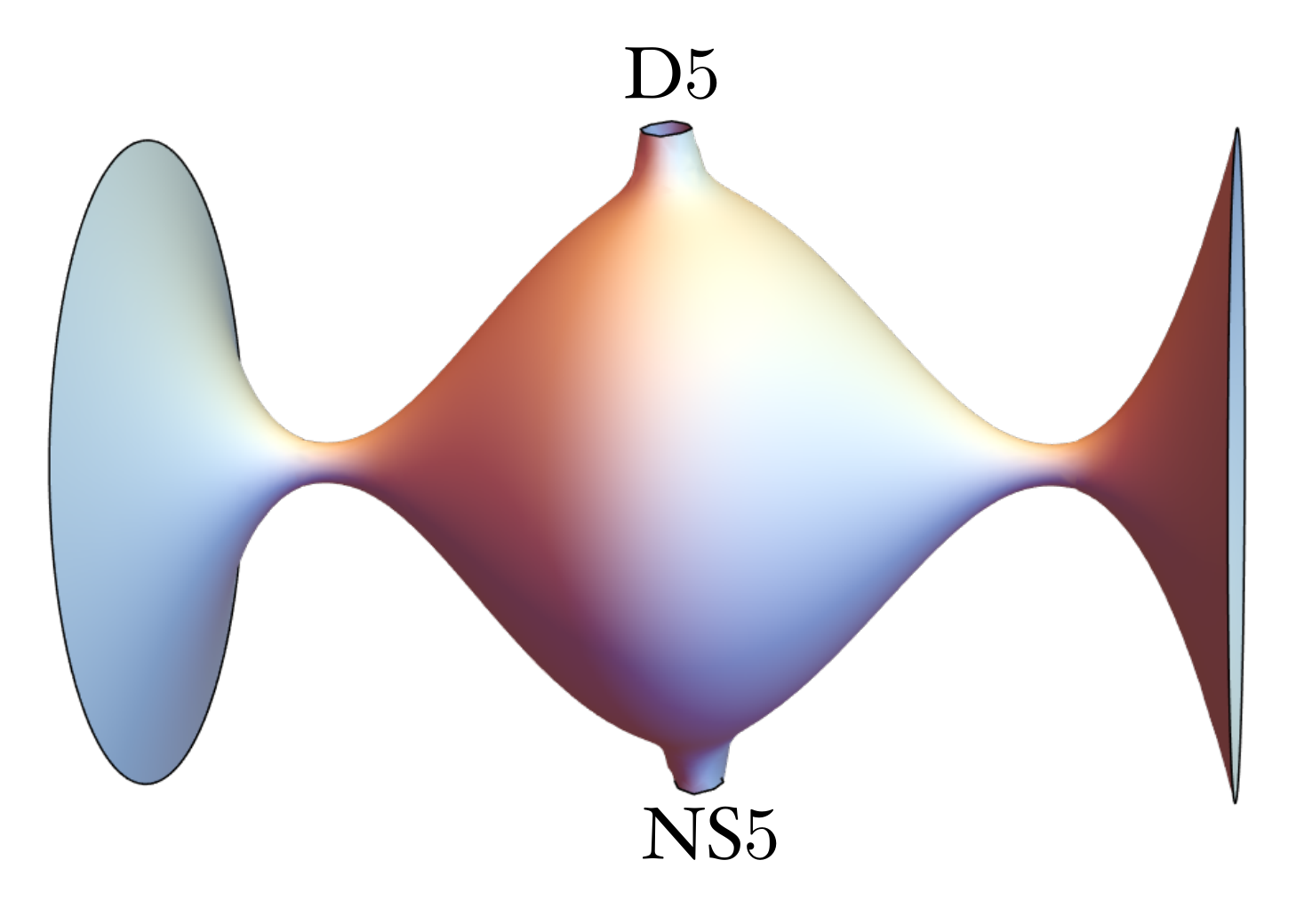}
	\caption{\small A cartoon of the internal space with two throats, similar to the KR model.}
	\label{fig:bl2}
\end{wrapfigure}

The simplest case, which includes no branes at all, was considered in \cite{bachas-estes}. Here the warping function is flat in a region $\{ x\in(-\delta\phi/2,+\delta\phi/2)\}$, and grows exponentially for larger $|x|$. (As its name suggests, $\delta \phi$ is related to the difference between the limits $\lim_{x\to \pm \infty} \phi(x)$.) In particular it does not have the peak near the origin that we saw in the KR model. A numerical analysis reveals that all the masses in the KK tower go to zero simultaneously as $\delta \phi\to \infty$, thus showing no sign of localization. We will see soon how to reproduce this in terms of Cheeger constants.

More general models with NS5- and D5-branes are much more promising \cite{assel-phd,bachas-estes,bachas-lavdas2}.
A cartoon is shown in Fig.~\ref{fig:bl2}: the space has a central ``bulb'' of volume $V$ and two tubes of length $\ell$ ending with two non-compact regions where the $S^5$ grows exponentially. This is similar to the AdS KR model: the central peak is replaced by the bulb, and the tubes serve to suppress higher-dimensional KK modes.

A simple version of this model was considered in \cite[App.~E]{assel-phd}, which we slightly generalize here. The harmonic functions read
\begin{equation}\label{eq:H1H2-2t}
	\begin{split}
	&H_1 = 2l_s^2\mathrm{Re}\left[-\ii a \sinh z- (N/4) \log\tanh\left(\frac{\ii\pi}4 -\frac z2\right)\right]\,,\\
	&H_2 = 2 l_s^2\mathrm{Re}\left[\hat a \cosh z - (\hat N/4)\log \tanh\left(\frac z2\right) \right]\,,
	\end{split}
\end{equation}
with $z = x+ \ii y$.
We take $a, \hat a \ll N, \hat N$. There are $N$ D5-branes and $\hat N$ NS5-branes at $x=0$. In the $x\to \pm \infty$ regions the metric is asymptotically AdS$_5\times S^5$, where both factors have curvature length $L_5$ given by $(L_5/l_s)^4 = 4\pi n \sim 8(a \hat N + \hat a N)$; $n$ is the $F_5$ flux quantum at $x\to +\infty$. The weighted volume of the bulb (the integral of $\ee^{8A}$ on it) is $V\sim  v (N \hat N)^2 l_s^8/L_4^2$, where $v\sim .05$ is a numerical factor and as usual $\Lambda_4 = - 3/L_4^2$.\footnote{\label{foot:L4}In the original papers, $L_4$ is set to $l_s$ by using invariance of (\ref{eq:wp}) under $A\to A+A_0$, $L_4\to \ee^{A_0} L_4$, $\dd s^2_{M_6}\to \ee^{-2A_0}\dd s^2_{M_6}$. Here we have restored it as an independent variable, but it is important to keep in mind that it is an ambiguous reference quantity; all physically meaningful quantities are measured with respect to it. For example, a product $mL_4$ is unambiguous.}
The tubes are in the regions $(\delta x -\frac12 \ell, \delta x +\frac12\ell)$ and $(-\delta x -\frac12 \ell, -\delta x +\frac12\ell)$, where
\begin{equation}\label{eq:dxl-2t}
	\delta x= \frac14\log\frac{N \hat N}{a \hat a} \, ,\qquad \ell =\frac12\left|\log\frac{a \hat N}{\hat a N}\right|\,.
\end{equation}
As usual in these models (for example the Janus case mentioned earlier), $\frac12\log\frac{a \hat N}{\hat a N}$ also happens to be equal to the difference $\delta \phi$ between the values of the dilaton at its two ends. 

A variant of this model with a single tube\footnote{This can be obtained for example from the models considered in \cite[Sec.~5.4]{deluca-deponti-mondino-t} by taking the limit $N_\mathrm{R}$, $\hat N_\mathrm{R}\to \infty$, keeping $\hat N_\mathrm{R}/N_\mathrm{R}$ constant.} can be thought of as a version of the KR model where the brane at the center is in fact a spacetime boundary. The models in \cite{dhoker-estes-gutperle2} have a larger number of tubes, and don't have a direct analogue in five dimensions.

Before we discuss the mass spectrum in these models, let us comment on the presence of a continuous part of the spectrum. As we anticipated in the general discussion in Section \ref{sub:genmetloc}, the presence of D5s and NS5s implies that a continuous spectrum does exist. However, the hypotheses in \cite[Th.~3.1]{cianchi-mazya} fail near the sources, but not for large distances.\footnote{The function $\mu(s): [0,\infty)\to [0,\infty)$ is the infimum of a certain quantity over sets of measure $\geqslant s$; $s/\mu(s)$ should go to zero as $s$ goes to both $0$ and $\infty$. For the spaces in this subsection, the limit fails as $s\to 0$ because of sets supported near the branes. For $s\to \infty$ we are checking a property that regards very large sets, involving the $x\to \pm \infty$ regions; this condition is satisfied, as can be shown using \cite[(3.7)]{cianchi-mazya}.} This suggests that the continuum starts at $r_0^{-1}\sim l_s^{-1} g_s^{-1/2}$, but not below. This is also supported by a numerical analysis. We will return to this later in this section.

The lightest spin-two mass $m_0$ was evaluated numerically in \cite[App.~E]{assel-phd} in the $a=\hat a$, $N= \hat N$ case, where it was found $m_0 L_4\sim O(a/N)$. It can be estimated more generally and precisely with the methods of \cite{bachas-lavdas,bachas-lavdas2}, which give
\begin{equation}\label{eq:m0-BL2}
	m_0^2 L_4^2\sim \frac{2^7 3 \pi^3 L_5^8}{V L_4^2} J(\ell)=\frac{2^{11}3 \pi^5}{v}\frac{n^2}{N^2 \hat N^2 } J(\ell)\,,
\end{equation}
where $J(\ell)= ( \ell \coth^3 \ell - \coth^2 \ell)^{-1}$ is of order one at $\ell=0$, and $O(\ell)^{-1}$ as $\ell\to \infty$.

The larger masses $m_{k>0}^2$ were not estimated in \cite{bachas-lavdas2}, but we can again do so using the Cheeger constants.
As a warm-up we evaluate $h_0$. We take $B_0$ of the form $(-x_0,x_0)$.  \eqref{eq:h0-inf} gives the estimate
\begin{equation}\label{eq:h0-BL0}
	h_0 \sim \inf_{x> \delta x -\ell/2}
	\frac1{\sqrt2 L_4}\frac{g(x - \delta x)^{3/2}}
	{\frac{4\pi^5}v (\frac{N \hat N}n)^2 + \int_{-\ell/2}^{x} \dd x' g(x'- \delta x)} \, ,\qquad g(x):= 1+ \frac{\cosh 2x}{\cosh \ell}\,,
\end{equation}
of the same form as in \cite[(5.22)]{deluca-deponti-mondino-t} for a similar (but compact) model in the same class. Recall that there is a tube of length $\ell$ and centered at $\delta x$. 
The requirement $x> \delta x-\ell/2$ is so that we stay away from the central bulb, where the back-reaction of the central NS5- and D5-branes is important; here the quantity to minimize is large and more complicated, although in principle computable exactly from (\ref{eq:H1H2-2t}). The function $g(x)$ is flat in the region $(-\ell/2, \ell/2)$, and grows exponentially for $x > \ell$; at $x\to \infty$, $M_6$ becomes non-compact. In finding the infimum in (\ref{eq:h0-BL0}), one wants to make $x$ grow almost all the way to $\delta x +\ell/2$, because the numerator stays almost constant while the denominator grows like $C+ x$. For $x> \delta x +\ell/2$, the numerator grows faster than the denominator. Recalling $L_5\sim n^{1/4}$, $V_1\sim (N \hat N)^2$, we get
\begin{equation}\label{eq:h0-BL}
	h_0 \sim \frac1{\sqrt2 L_4}\left(\frac{4\pi^5}v \frac{N^2 \hat N^2}{n^2} + \ell\right)^{-1}\,.
\end{equation}
A more precise estimate can be obtained as in \cite{deluca-deponti-mondino-t} using Lambert functions. Unfortunately, as also remarked there, the general upper bounds
in \eqref{eq:bound0InfVol} are not useful for non-compact models of this type, but the lower bound $m_0 \geqslant h_0/2$  does apply, and is indeed compatible with it.

We can now try to give a similar estimate for the second-lightest mass $m_1$, by using \eqref{eq:bound1InfVol}. To do so, we need to estimate $h_1$ from \eqref{eq:h1-inf} and we take $B_0=(-\delta x + \ell/2,x_0)$, $B_1=(x_0,x_1)$. It is again strongly favored to take both $x_0$, $x_1$ in the tube region $(-\ell/2 + \delta x, \ell/2 + \delta x)$. We have  
\begin{equation}\label{eq:h1-BL}
	\frac{\text{Vol}_A(\partial B_0)}{\text{Vol}_A(B_0)}\sim \frac2{\frac{8\pi^5}v (\frac{N \hat N}n)^2+ x_0-\delta x + \frac32\ell} \, ,\qquad \frac{\text{Vol}_A(\partial B_1)}{\text{Vol}_A(B_1)}\sim \frac1{x_1-x_0}\,.
\end{equation}

The largest of these two is always $\text{Vol}_A(\partial B_1)/\text{Vol}_A(B_1)$. So we need to take the infimum of this, which is obtained by making $x_1-x_0=\ell$, or in other words by making $B_1$ extend as long as the tube. So we conclude $h_1 \sim \frac1{\ell}$.

We would like now to use the lower bound in \eqref{eq:bound1InfVol}. Unfortunately, this comes from theorem \cite[Th.~6.8]{deluca-deponti-mondino-t-entropy}, which needs a hypothesis on the growth of the volumes of balls that fails in the asymptotic regions $x\to \pm \infty$. (An earlier version \cite[Th.~4.9]{deluca-deponti-mondino-t} requires a lower Ricci bound that also fails in our case.) To overcome this difficulty, we notice that for large $|x|$ the metric behaves as $\sim \ee^{-2x}(\dd s^2_{\mathrm{AdS}_4}+ \dd x^2 + \dd s^2_{S^5})$, and the normalizable eigenfunctions decay exponentially as $\psi(x,y)\sim e^{(-2-\sqrt{4+m^2})|x|}f_m(y)$, with $\Delta_{S^5} f_m = m^2 f_m$. We can thus approximate very well the full problem with the Neumann problem obtained by cutting off the geometry to a region $M_{\bar x}:= M_6 \cap[-\bar x,\bar x]$ for a large enough $\bar x$. We have also checked numerically that above a certain threshold the spectrum is insensitive to the cutoff $\bar x$.

Working in $M_{\bar x}$ gives us that as anticipated the spectrum is discrete, at least below the D5 threshold $l_s^{-1} g_s^{-1/2}$. Moreover, the cut off model now satisfies the hypotheses in \cite[Th.~6.8]{deluca-deponti-mondino-t-entropy}. But there is a last subtlety for us to tackle: the Neumann eigenvalues on the cut off model include two additional light modes, which are spurious in that the corresponding eigenfunctions go to zero when the cutoff is removed. To see these spurious eigenvalues $\bar m_0$, $\bar m_1$, let us see how our previous discussion of the Cheeger constants is modified by the cutoff. (We denote quantities relative to the cut off model by a bar.) First, $\bar h_0$ is realized by taking $\bar B_0=M_{\bar x}$. Since $\partial \bar B_0\subset \partial M_{\bar x}$, $\mathrm{Vol}_A(\partial \bar B_0)=0$ and $\bar h_0=0$, which by (\ref{eq:h0-inf}) implies $\bar m_0=0$. Next, $\bar h_1$ is realized by taking $\bar B_0=(-\bar x, - \delta x -\ell/2)$ and $\bar B_1=(\delta x + \ell/2 , \bar x)$, namely the regions beyond the two tubes in Fig.~\ref{fig:bl2}. When $\bar x $ gets large, $\mathrm{Vol}_A(\partial \bar B_i)$ remain fixed while $\mathrm{Vol}_A(\bar B_i)$ diverges; in agreement with (\ref{eq:bound1InfVol}), $\bar m_1\to 0$. The next two eigenvalues are not spurious and survive in the non-compact limit $\bar x\to0$; the collections of $B_i$ are obtained by taking those we discussed around (\ref{eq:h0-BL}) and (\ref{eq:h1-BL}), and adding to them the $\bar B_0$, $\bar B_1$ we just saw. In particular, the $m_1$ for the non-compact model is the limit of $\bar m_3$ for the cutoff model; these are related to $h_1$ and $\bar h_3$ respectively. We can now use \cite[Th.~6.8]{deluca-deponti-mondino-t-entropy} for $k=3$ to obtain
\begin{equation}\label{eq:m1-BL2}
	m_1^2 L_4^2> \frac{c}{\ell^2}
\end{equation}
for a certain universal constant $c$.

When $\ell$ is large, there is no separation between $m_1$ and $m_0$. This is similar to what was found in the aforementioned pure Janus model in \cite{bachas-estes}. In that case, the bulb is absent, and the tubes join to form a single one. The computation of the Cheeger constants is even clearer than in the earlier case: $h_0$ is realized as $\text{Vol}_A(\partial B)/\text{Vol}_A(B)$ for a  $B$ that fills the entire tube; $h_1$ by considering $B_0$, $B_1$ that each take half the tube. So $h_0\sim 1/ \ell$, $h_1 \sim 1/ 2\ell$. (A generalization to a higher number of $B_i$ exists, and gives $h_k\sim 1/(k+1) \delta \phi$ for similar reasons.) 

However, the tube can also be made short: looking at (\ref{eq:dxl-2t}), this happens for example when $a=\hat a$, $N=\hat N$ (or more generally when they are almost equal). Now $h_1$ is no longer small. We expect it to be $O(1)$, and $m_1\sim 1/L_4$. But $m_0$ can be made small by taking
\begin{equation}
	n\ll N \hat N= N^2.
\end{equation}
So in this regime, to which we restrict our attention from now on, we do have a hierarchy:
\begin{equation}
	\frac{m_0}{m_1}\sim \frac{n}{N^2} \ll 1 \,.
\end{equation}
This indicates that we can have gravity localization in the two-throat model.

Similar to our discussion of the KR model, by itself this hierarchy only implies that gravity behaves in a four-dimensional way for scales $L_4 \ll R \ll (N^2/n) L_4$. To see if this persists below the cosmological scale, we need information about the wave-functions themselves. In the RS model, the suppression of $\psi_{k>0}$ on the brane was responsible for the absence of a $R^{-2}$ term in the potential, as discussed below (\ref{eq:massive-V}); in the KR model, it makes it negligible, so that localization persists up to $R \sim L_5$, well below $L_4$.

Our mathematical theorems do not put constraints on the $\psi_k$, but we will try to proceed anyway. In this class of backgrounds, there are two types of modes to consider: 
\begin{itemize}
	\item For matter modes localized on the NS5-branes, numerically we find $\psi_k(x=0,y=0)^2\sim N^{-4}$ for the first few $k >0$. However, we find the same behavior with $N$ also for $k = 0$, in contrast to the KR case, for which we observed above (\ref{eq:V-KR}) that higher eigenfunctions $\psi_{k>0}$ are suppressed at the origin relative to $\psi_0$. Now (\ref{eq:massive-V}) gives, with a similar logic to that in footnote \ref{foot:quickV}:
\begin{equation}
	V\sim G M_1 M_2 \left(\frac1R + \frac{L_4}{R^2} \right)\,.
\end{equation}
Thus for these modes localization only really works up to the cosmological scale, $R>L_4$.\footnote{Actually for $R>L_4$ (4.21) will change, with a large contribution from the gravitational potential of AdS itself. The overall leading behavior will still appear four-dimensional.} The result is similar for D5-branes.
	\item Further four-dimensional matter particles originate from closed-string modes whose KK wave-function $\lambda$ is concentrated on the bulb, which we expect to provide the ${\mathcal N}=4$ superpartners of the lowest spin-two mode. For such modes, the same discussion leading to (\ref{eq:V-hope}) would apply.

\item{Finally, other possibilities might arise by extending this class of backgrounds with the inclusion of branes that support four-dimensional matter and are localized at an internal locus where the higher eigenfunctions are suppressed compared to $\psi_0$. This would require to recompute the supergravity solution, as their backreaction can be severe.
}

\end{itemize}
 
In summary, both for the KR and for the present model the lowest mass $m_0$ is very small, and the next $m_1\sim 1/L_4$. However, for the KR model the suppression of the higher eigenfunctions $\psi_{k>0}$ on the brane pushes the KK length scale to $L_5$. In contrast, the models in this section don't seem to display this suppression, at least for the modes localized on the D5-branes and NS5-branes, so the KK length scale is $L_4$; for some closed-string modes it might be possible to make this somewhat lower. 

One might also wonder if it is possible to push directly $m_1$ to a higher value than $1/L_4$.  In this particular case, the bound of Prop.~\ref{prop:Sep} becomes vacuous since in the asymptotic AdS$_5\times S^5$ regions $|\dd A|$ is not bounded. We can again consider the cutoff model, where the asymptotic region is removed, and apply the bound of Prop.~\ref{prop:Sep} to the cut off space. In this case, however, the cutoff cannot be removed, as it controls both the value of $K$ and the precision of the approximation.


\section*{Acknowledgements}

We thank B.~Assel, C.~Bachas, S.~Giri, A.~Karch, L.~Randall and A.~Zaffaroni for discussions. The authors are grateful to L.~Dello Schiavo for several suggestions.  GBDL is supported in part by the Simons Foundation Origins of the Universe Initiative (modern inflationary cosmology collaboration) and by a Simons Investigator award. AM is supported by the ERC Starting Grant 802689 ``CURVATURE''. AT is supported in part by INFN and by MIUR-PRIN contract 2017CC72MK003.

\bibliography{at}
\bibliographystyle{at}

\end{document}